\documentclass[12pt]{article}
\usepackage[margin=1in]{geometry}
\usepackage[utf8]{inputenc} 
\usepackage[T1]{fontenc}
\usepackage{times}

\usepackage[normalem]{ulem}

\usepackage[misc]{ifsym}
\usepackage{balance} 
\usepackage{tikz}
\usepackage{booktabs}
\usepackage{footnote}
\usepackage{amsmath,amssymb,amsthm,amsfonts,latexsym,bbm,xspace,graphicx,float,mathtools,dsfont}
\usetikzlibrary{arrows}
\newtheorem{remark}{Remark}
\usepackage{multirow}

\usepackage{natbib}

\newtheorem{innercustomthm}{Theorem}
\newtheorem{definition}{Definition}
\newtheorem{theorem}{Theorem}
\newtheorem{proposition}{Proposition}
\newtheorem{lemma}{Lemma}
\newtheorem{corollary}{Corollary}
\theoremstyle{definition}
\newtheorem{example}{Example}
\newtheorem{innercustomlem}{Lemma}
\newenvironment{customlem}[1]
  {\renewcommand\theinnercustomlem{#1}\innercustomlem}
  {\endinnercustomlem}
\newenvironment{customthm}[1]
  {\renewcommand\theinnercustomthm{#1}\innercustomthm}
  {\endinnercustomthm}

  \usepackage{authblk}
  \begin{document}
\title{Proportional Fairness in Obnoxious Facility Location}

\author[1(\Letter )]{Alexander Lam\thanks{alexlam@cityu.edu.hk}}
\author[2]{Haris Aziz}
\author[3]{Bo Li}
\author[2]{Fahimeh Ramezani}
\author[2]{Toby Walsh}
\affil[1]{City University of Hong Kong}
\affil[2]{UNSW Sydney}
\affil[3]{Hong Kong Polytechnic University}
\date{}

\maketitle

\begin{abstract}
We consider the obnoxious facility location problem (in which agents prefer the facility location to be far from them) and propose a hierarchy of distance-based proportional fairness concepts for the problem. These fairness axioms ensure that groups of agents at the same location are guaranteed to be a distance from the facility proportional to their group size. We consider deterministic and randomized mechanisms, and compute tight bounds on the price of proportional fairness. In the deterministic setting, we show that our proportional fairness axioms are incompatible with strategyproofness, and prove asymptotically tight $\epsilon$-price of anarchy and stability bounds for proportionally fair welfare-optimal mechanisms. In the randomized setting, we identify proportionally fair and strategyproof mechanisms that give an expected welfare within a constant factor of the optimal welfare. Finally, we prove existence results for two extensions to our model.
\end{abstract}

\section{Introduction}
In the \emph{obnoxious facility location problem (OFLP)}, some undesirable facility such as a garbage dump or an oil refinery is to be located on a unit interval (i.e. the domain of locations is $[0,1]$), and the agents along the interval wish to be as far from the facility as possible~\citep{FLS+20,CYZ11a,IbNa21a,CHYZ19}. In this problem, agents have single-dipped preferences, contrasting with the single-peaked preferences of agents in the classic facility location problem (in which agents prefer to be located as close as possible to the facility). 

The obnoxious facility location problem models many real-world facility placements which negatively impact nearby agents, such as a prison or a power plant~\citep{ChDr22}. Aside from the geographic placement of an obnoxious facility, the OFLP can also be applied to various collective decision making problems. For instance, when agents are averse to their worst possible social outcomes (represented by their locations), the problem captures issues where a decision needs to be made on a social policy or a budget composition. When a socially sensitive or a politically undesirable policy needs to be implemented, the placements of such a policy in the space of outcomes may need to take equity considerations.

It is known that placing the facility at one of the interval endpoints maximizes the sum of agent distances \citep{CYZ13}, but such a solution may not be `proportionally fair' for the agents. To build intuition, consider the instance depicted in Figure~\ref{fig:example}. The optimal utilitarian solution (which maximizes the sum of agent distances) places the facility at $0$, disproportionately disadvantaging the agents at $0.1$ who are located only $0.1$ distance from the facility. A facility location of $0.45$ results in both groups of agents having the same distance from the facility, and would be considered to be more `fair' in the egalitarian sense. However, it is not proportionally fair: despite having over twice as many agents, the group of agents at $0.8$ have the same distance from the facility as the group of agents at $0.1$. A proportionally fair solution places the facility at $0.3$, and results in the distance between a group of agents and the facility being proportional to the size of the group.

\begin{figure}[h!]
 	  	 \begin{center}  
   		      	             \begin{tikzpicture}[scale=1.0]
   			   	      	                 \centering
   	   	      	                 \draw[-] (0,0) -- (10,0);
								 
   								   \draw[-] (0,0) -- (0,0.25);
   								    \draw[-] (1,0) -- (1,0.25);
   									    \draw[-] (2,0) -- (2,0.25);
   										 \draw[-] (3,0) -- (3,0.25);
   										  \draw[-] (4,0) -- (4,0.25);
    \draw[-] (5,0) -- (5,0.25);
     \draw[-] (6,0) -- (6,0.25);
     \draw[-] (7,0) -- (7,0.25);
     \draw[-] (8,0) -- (8,0.25);
     \draw[-] (9,0) -- (9,0.25);
     \draw[-] (10,0) -- (10,0.25);
  
     \draw (0,-.4) node(c){\small $0$};
       \draw (1,-.4) node(c){\small $0.1$};
   	    \draw (2,-.4) node(c){\small $0.2$};
   		    \draw (3,-.4) node(c){\small $0.3$};
   			 \draw (4,-.4) node(c){\small $0.4$};
   			 	 \draw (5,-.4) node(c){\small $0.5$};
   			 	 \draw (6,-.4) node(c){\small $0.6$};
   			 	 \draw (7,-.4) node(c){\small $0.7$};
   			 	 \draw (8,-.4) node(c){\small $0.8$};
   			 	 \draw (9,-.4) node(c){\small $0.9$};
       \draw (10,-.4) node(c){\small $1$};
	
   	  \draw (1,0.6) node(c){\small $\text{x}$};
   	  \draw (1,0.8) node(c){\small $\text{x}$};

   				   \draw (8,0.6) node(c){\small $\text{x}$};
   				   	   \draw (8,0.8) node(c){\small $\text{x}$};
   				   	   \draw (8,1) node(c){\small $\text{x}$};
   				   	   \draw (8,1.2) node(c){\small $\text{x}$};
   				   	   \draw (8,1.4) node(c){\small $\text{x}$};
   						   \draw (0,1.5) node(c){\small $f^*_{UW}$};
   						   \draw (3,1.5) node(c){\small 2-UFS};
   						    \draw (4.5,1.5) node(c){\small $f^*_{EW}$};
   							\draw (0,0) node(c)[circle,fill,inner sep=1.5pt]{};
   							\draw (4.5,0) node(c)[circle,fill,inner sep=1.5pt]{};
   							\draw (3,0) node(c)[circle,fill,inner sep=1.5pt]{};
   			   	      	  \end{tikzpicture}
   			   	       	\end{center}
   			   	      	 \caption{OFLP with agent location profile $(0.1,0.1,0.8,0.8, 0.8,0.8)$ represented by \textmd{x}. The facility locations (represented by \textbullet) correspond to a utilitarian outcome, $f^*_{UW}=0$; a proportionally fair outcome, $\textmd{2-UFS}=0.3$; and an egalitarian outcome, $f^*_{EW}=0.45$.}
   			   	      	\label{fig:example}
   			   	      	\end{figure}
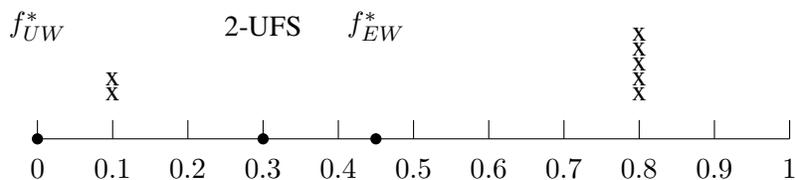

In this work, we pursue notions of \emph{proportional fairness} as a central concern for the problem. Specifically, we formulate a hierarchy of proportional fairness axioms which guarantee that each group of agents at the same location are a distance from the facility proportional to the relative size of the group. While proportional fairness axioms have been formulated and studied in the classic facility location problem \citep{ALLW21}, they have not yet been applied to the OFLP. Our paper provides a comprehensive overview of proportionally fair solutions for the obnoxious facility location problem, examining the interplay between proportional fairness and utilitarian/egalitarian welfare, and investigating concerns of agent strategic behaviour in both the deterministic and randomized settings.

\paragraph{Contributions}
\begin{itemize}
\item We formalize (approximate) proportional fairness concepts such as 2-Individual Fair Share (2-IFS) and 2-Unanimous Fair Share (2-UFS) in the context of the obnoxious facility location problem. Several of the definitions are natural adaptations of axioms from fair division and participatory budgeting. 

\item We find tight bounds on the price of 2-IFS and 2-UFS fairness for the objectives of egalitarian and utilitarian welfare, in both the deterministic and randomized settings.

\item We prove that our proportional fairness axioms are incompatible with strategyproofness in the deterministic setting, and give strategyproof randomized mechanisms that satisfy these proportional fairness axioms in expectation and either have a constant approximation ratio for utilitarian welfare or are optimal for egalitarian welfare.

\item For the deterministic mechanisms that maximize utilitarian welfare under the constraints of 2-IFS and 2-UFS, we prove that a pure $\epsilon$-Nash equilibrium always exists. We then find asymptotically tight linear bounds on the corresponding $\epsilon$-price of anarchy, as well as asymptotically tight constant bounds on the corresponding $\epsilon$-price of stability.

\item Finally, we give two possible extensions of our model: the fairness axiom of 2-Proportional Fairness (2-PF), which is stronger than 2-UFS as it captures proportional fairness concerns for groups of agents near but not necessarily at the same location, and the hybrid model, which also includes `classic' agents who instead want to be near the facility. We prove existence results for both extensions. 
\end{itemize}

Table~\ref{tab:results} summarizes some of our results. Results lacking proofs are proven in the appendix.
  \begin{table}[]
   			   	      	\centering
 			    			
 \caption{Price of fairness and welfare approximation results.}
\begin{tabular}{lllll}
\cline{3-5}
                               &                      & \multicolumn{2}{l}{Price of Fairness} & \multirow{2}{*}{\begin{tabular}[c]{@{}l@{}}Best Known\\ 2-UFS SP Approx.\end{tabular}} \\
                               &                      & 2-IFS            & 2-UFS              &                                                                                                     \\ \hline
\multirow{2}{*}{Deterministic} & Utilitarian Welfare & 2 (Thm.~\ref{thm: 2-IFS UW})               & 2 (Thm.~\ref{thm: 2-UFS UW})                 & \textbf{Incompatible}                                                                                        \\ \cline{3-4}
                               & Egalitarian Welfare & 1 (Prop.~\ref{prop:Po2IFSEW})               & $n-1$ (Thm.~\ref{thm: 2-UFS EW})                & (Prop.~\ref{prop:detnofair})                                                                                                 \\ \hline
\multirow{2}{*}{Randomized}    & Utilitarian Welfare & 12/11 (Cor.~\ref{cor:PoF2IFSUW})           & 1.094$\dots$ (Cor.~\ref{cor:PoF2UFSUW})            & 1.5 (Thm.~\ref{thm: mech2})                                                                                                 
                            \\ \cline{3-5} 
                               & Egalitarian Welfare & 1 (Prop.~\ref{prop:Po2IFSEW})               & 1 (Cor.~\ref{cor:PoF2UFSEW})                  & 1 (Prop.~\ref{prop:randEWSP})                                                                                                                                                                                                   \\
  \hline
\end{tabular}

 \label{tab:results}
\end{table}

  \paragraph{Related Work}

The papers most relevant to our research are those that treat the facility as obnoxious: agents prefer the facility to be as far from them as possible. Similar to the classical facility location problem, early operations research on the OFLP apply an optimization approach to compute solutions; a review of these approaches is given by \citet{ChDr22}.  There are several papers on the obnoxious facility location problem that apply a mechanism design approach, assuming agents' location are private information. This was initiated by \citet{CYZ11a,CYZ13}, who define an agent's utility as its distance from the facility, and design strategyproof mechanisms which approximate the optimal utilitarian welfare on the path and network metrics. Other recent examples of related papers include~\citep{CHYZ19,FLS+20,IbNa21a,XLLD21}. These papers do not pose or study the fairness concepts that we explore in this paper. 

Notions of fairness in various collective decision problems have been widely explored over the last few decades \citep{Moul03,Nash50,Shap53}. Fairness objectives specifically relevant to the facility location problem include maximum cost/egalitarian welfare (see, e.g. \citep{PrTe13,WWLC21}) and maximum total/average group cost \citep{ZLC22}. Rather than optimize/approximate fairness objectives, we focus on solutions enforcing proportional fairness axioms, in which groups of agents with similar or identical preferences/locations) have a minimum utility guarantee relative to the group size. The axioms of proportional fairness that we present stem from several related areas of social choice. Individual Fair Share (IFS) is closely related to the axiom of proportionality proposed by \citet{Stei48a}, and appears in participatory budgeting along with Unanimous Fair Share (UFS) \citep{BMS05,ABM19}. Asll of our proportional fairness axioms have been studied in the classical facility location problem by \citet{ALLW21}.
 
In our paper, we also analyse the loss of efficiency, defined as the price of fairness, of implementing the proportional fairness axioms that we have proposed. The price of fairness has been studied for some variations of the facility location problem, such as when there is a lexicographic minimax objective \citep{BKJ14}, or when facilities have preferences over subsets of agents, and the fairness is observed from the facilities' perspectives \citep{WaZh21}. Many recent results on price of fairness have also been found in various social choice contexts, such as fair division and probabilistic social choice \citep{BBS20,CKKK12,BLMS21,BFT11a}.

As strategyproofness is impossible in our deterministic setting, we present results on the existence of pure Nash equilibria, and the prices of anarchy and stability. Similar models where such results are proven include a variation of the Hotelling-Downs model \citep{FFO16}, and two-stage facility location games where both facilities and clients act strategically \citep{KLMS21}. In the classic facility location problem, \citet{ALLW21} characterize the pure Nash equilibria of strictly monotonic facility location mechanisms satisfying UFS and show that the resulting equilibrium facility location is also guaranteed to satisfy UFS. For certain mechanisms in our setting, a pure Nash equilibrium may not exist, so we prove the existence of the approximate \emph{pure $\epsilon$-Nash equilibrium}. Examples of papers applying this notion to other settings include \citep{ChSi11,MSA15}.

The second half of our paper focuses on the randomized setting to overcome the incompatibility with strategyproofness. The use of randomized mechanisms to overcome impossibility results is prevalent in many social choice contexts (see, e.g., \citep{Bran17,Aziz19}). Additionally, \citet{ALSW22} use a randomized approach in the classic facility location problem to achieve stronger notions of proportional fairness, providing a unique characterization of universally anonymous and universally truthful mechanisms satisfying an axiom called Strong Proportionality. The use of randomized mechanisms also results in better approximation ratio/price of fairness bounds. This is common in many variants of the facility location problem, such as when agents have fractional or optional preferences \citep{FLL+18,CFL+20}, or in the hybrid facility location model \citep{FLS+20}.

\section{Model}
Let $N=\{1, \ldots, n\}$ be a set of agents, and let $X:=[0,1]$ be the domain of locations.\footnote{Our results can be naturally extended to any compact interval on $\mathbb{R}$.} Agent $i$'s location is denoted by $x_i\in X$; the profile of agent locations is denoted by $x=(x_1, \ldots, x_n)\in X^n$. We also assume the agent locations are ordered such that $x_1\leq \dots \leq x_n$.
A \emph{deterministic mechanism} is a mapping $f\ : \ X^n\rightarrow X$ from a location profile $\hat{x}\in X^n$ to a facility location $y\in X$. Given a facility location $y\in X$, agent $i$'s utility\footnote{This definition is consistent with \citep{CYZ13}.} is equal to its distance from the facility $u(y,x_i):=|y-x_i|$. We are interested in maximizing the objectives of \emph{Utilitarian Welfare} (UW), defined for a facility location $y$ and location profile $x$ as the sum of agent utilities $\sum_i u(y,x_i)$, and \emph{Egalitarian Welfare} (EW), defined as the minimum agent utility $\min_i u(y,x_i)$.

Note that the preferences in OFLP can be viewed as \emph{single-dipped}, contrasting with the \emph{single-peaked preferences} of the classical facility location problem (FLP). The underlying model of both FLP and OFLP is the same except that the agents' preferences have a different structure. Unless specified otherwise, we will state results for the obnoxious facility location problem (OFLP).

\section{Proportional Fairness Axioms}
In this section, we introduce proportional fairness axioms for the obnoxious facility location problem.
				\subsection{Individual Fair Share}
			We first present an adaptation of Individual Fair Share (IFS), the weakest of our proportional fairness axioms (as studied by \citet{ALLW21} in the context of the classic facility location problem). IFS provides a minimum distance guarantee between each agent and the facility, requiring that each agent has at least $\frac{1}{n}$ utility. By placing two agents at $\frac{1}{4}$ and $\frac{3}{4}$, it is easy to see that an IFS solution may not exist. As a result, we turn to approximations of IFS.
				\begin{definition}[$\alpha$-Individual Fair Share (IFS)]\label{def: IFS}
				Given a profile of locations $x$, a facility location $y$ satisfies \emph{$\alpha$-Individual Fair Share ($\alpha$-IFS)} if $$u(y, x_i)\ge \frac{1}{\alpha n} \quad \forall i\in N.$$
				\end{definition}
				
		We find that the lowest value of $\alpha$ such that an $\alpha-$IFS solution always exists is $\alpha=2$. Intuitively, with $\alpha=2$, each agent has an open interval of radius $\frac{1}{2n}$ around its location. The sum of interval lengths is $1$, meaning there will always be a 2-IFS solution. For any $\alpha<2$, an $\alpha-$IFS solution may not always exist as the sum of interval lengths will exceed 1.
				\begin{proposition}\label{prop2IFS}
					The lowest value of $\alpha$ for which an $\alpha$-IFS solution always exists is $\alpha=2$.
					\end{proposition}
				\begin{proof}
				Consider $n$ agents at ordered locations $x_1,\dots,x_n$. For each agent $i$, we construct an open interval $B_i$ with center $x_i$ and radius $\frac{1}{2n}$: $B_i=\{z|d(z,x_i)<\frac{|1|}{2n}\}$. Note that the sum of interval lengths is $1$.

There are two cases:
\begin{itemize}
\item $B_i\cap B_j=\emptyset$ for all $i\neq j$. As the sum of interval lengths is $1$, the boundaries of two consecutive intervals intersect, and thus the facility can be placed at the boundary of any interval.
\item $B_i\cap B_j\neq \emptyset$ for some $i\neq j$. In this case, the length of $B_1\cap \dots \cap B_n$ is less than $1$, hence there must be a region on $[0,1]$ that is not covered by any open interval. The facility can be placed within this region to achieve a $2-$IFS solution.
\end{itemize}
To see that an $\alpha-$IFS solution may not exist for $\alpha<2$, consider for $n=2$ the location profile $(\frac{1}{4},\frac{3}{4})$, in which the intersection of the intervals encompasses the entire unit interval.										
\end{proof}

A polynomial time 2-IFS mechanism (which we denote as $f^*_{2IFS}$) that maximizes the utilitarian welfare simply iterates through the endpoints of the intervals which satisfy the constraint and outputs the optimal facility location, breaking ties in favour of the leftmost optimal location.
		\subsection{Unanimous Fair Share}
		We now present Unanimous Fair Share (UFS), a strengthening and generalization of IFS to groups of agents at the same location. Informally, if there are $k$ agents at the same location, then UFS requires that the facility is placed at least $\frac{k}{n}$ distance from these agents. Under UFS, agents are not considered to be in the same group if they are very close but not exactly co-located. However, the co-location of agents often naturally arises in practice, such as when multiple citizens live in the same apartment building, or when considering populations of towns. Towards the end of the paper, we propose a stronger proportional fairness axiom which considers agents at near but not necessarily the same location to be part of the same group.
		
		Again, we focus on approximations of UFS as a UFS solution may not exist.
		
		\begin{definition}[$\alpha$-Unanimous Fair Share (UFS)]\label{def: UFS}
			Given a profile of locations $\boldsymbol{x}$,  a facility location $y$ satisfies \emph{$\alpha$-Unanimous Fair Share ($\alpha$-UFS)} if for any set of agents $S$ with identical location, $$u(y, x_i)\ge \frac{|S|}{\alpha n} \quad \forall i\in S.$$
		\end{definition}
		
		Note that $\alpha-$UFS implies $\alpha-$IFS. As with $\alpha-$IFS, we find that the optimal value of $\alpha$ for which an $\alpha$-UFS solution always exists is $\alpha=2$. The proof intuition is similar to that of Proposition~\ref{prop2IFS}, but the intervals vary in size depending on the number of agents in the group.

			\begin{proposition}\label{optUFS}
The lowest value of $\alpha$ for which an $\alpha$-UFS solution always exists is $\alpha=2$.
\end{proposition}
\begin{proof}
Consider $n$ agents at $m$ unique ordered locations $x_1,\dots,x_m$, and for $i\in [m]$, let $S_i$ denote the group of agents at location $x_i$. For each $S_i$, we construct an open interval $B_i$ with center $x_i$ and radius $\frac{|S_i|}{2n}$: $B_i=\{z|d(z,x_i)<\frac{|S_i|}{2n}\}$. Note that the sum of interval lengths is $\sum^m_{i=1}\frac{|S_i|}{n}=1$.

There are two cases:
\begin{itemize}
\item $B_i\cap B_j=\emptyset$ for all $i\neq j$. As the sum of interval lengths is $1$, the boundaries of two consecutive intervals intersect, and thus the facility can be placed at the boundary of any interval.
\item $B_i\cap B_j\neq \emptyset$ for some $i\neq j$. In this case, the length of $B_1\cap \dots \cap B_m$ is less than $1$, hence there must be a region on $[0,1]$ that is not covered by any interval. The facility can be placed within this region to achieve a $2-$UFS solution.
\end{itemize}
To see that an $\alpha-$UFS solution may not exist for $\alpha<2$, place all $n$ agents at location $\frac{1}{2}$.
\end{proof}

			Similar to $f^*_{2IFS}$, a polynomial time mechanism (which we denote as $f^*_{2UFS}$) that computes the optimal 2-UFS facility location for utilitarian welfare iterates through the endpoints of the intervals satisfying 2-UFS and outputs the optimal facility location, breaking ties in favour of the leftmost optimal location.
			
\section{Deterministic Setting}
We begin with the deterministic setting, analyzing the price of proportional fairness and agent strategic behaviour. All results stated in this section are for the deterministic setting.		
\subsection{Price of Fairness}
			In this section, we analyze the price of fairness for our (approximate) fairness axioms.\footnote{The price of fairness can also be interpreted as the approximation ratio for the respective optimal mechanism satisfying the fairness constraint.} Informally, the price of fairness measures the loss of efficiency from imposing a certain fairness constraint. We focus on the objectives of utilitarian and egalitarian welfare, defined as the sum of utilities and the minimum agent utility, respectively. 
			
			A \emph{fairness property} $P$ is a mapping from an agent location profile $x\in X^n$ to a (possibly empty) set of facility locations $P(x)\in X$. Every facility location $P(x)$ satisfies the fairness property $P$. The price of fairness for property $P$ is the worst-case ratio between the optimal welfare and the optimal welfare from a facility location satisfying $P$.
			\begin{definition}	 [Price of Fairness for Utilitarian/Egalitarian Welfare]
		
Let \{$f^*_{UW}$,$f^*_{EW}\}$ be the mechanism that returns the solution maximizing utilitarian/egalitarian welfare. For Utilitarian/Egalitarian Welfare and fairness property $P$, we define the price of fairness as the worst-case ratio (over all location profiles) between the optimal Utilitarian/Egalitarian Welfare and the optimal Utilitarian/Egalitarian Welfare achieved by a facility location satisfying $P$:
$$\max_{x\in [0,1]^n}\frac{W(f^*(x),x)}{\max_{y\in P(x)}W(y,x)}.$$
For Utilitarian Welfare, $f^*(x):=f^*_{UW}(x)$ and $W(y,x):=\sum_i u(y,x_i)$.

\noindent For Egalitarian Welfare, $f^*(x):=f^*_{EW}(x)$ and $W(y,x):=\min_i u(y,x_i)$.
\end{definition}
			 \subsubsection{Utilitarian Welfare}
		The utilitarian welfare (UW) of an instance is a standard measure of efficiency. Finding the price of our proportional fairness axioms for utilitarian welfare quantifies the impact on efficiency when the OFLP system is constrained to be proportionally fair.
		
We now move to compute the prices of 2-IFS and 2-UFS fairness for utilitarian welfare. Recall that the solution maximizing utilitarian welfare must be either $0$ or $1$ \citep{CYZ13}. To prove the price of fairness lower bounds, we place the agents such that the only feasible 2-IFS/UFS solution lies in the optimal median interval (see, e.g. Figure~\ref{fig:UW}).

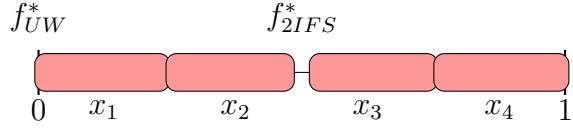
\begin{figure}
\centering
\begin{tikzpicture}
			\draw[] (-3.5,0) -- (3.5,0); 
			\draw [thick] (-3.5,0.3)  -- (-3.5,-0.3);  
			\draw [thick] (3.5,0.3)  -- (3.5,-0.3); 
		
			\node at (-3.5, -0.5) {$0$}; 
			\node at (3.5, -0.5) {$1$}; 
			\node at (-2.63, -0.5) {$x_1$}; 
			\node at (-0.88, -0.5) {$x_2$}; 
			\node at (0.88, -0.5) {$x_3$}; 
			\node at (2.63, -0.5) {$x_4$}; 
		
			\filldraw[fill=red!40, rounded corners](-3.55,-0.25) rectangle (-1.755,0.25); 
			\node at (-3.5, 0.73) {$f^*_{UW}$}; 
			\node at (0,0.73) {$f^*_{2IFS}$};
			\filldraw[fill=red!40, rounded corners](-1.805,-0.25) rectangle (-0.1,0.25); 
			\filldraw[fill=red!40, rounded corners](0.1,-0.25) rectangle (1.805,0.25);
			\filldraw[fill=red!40, rounded corners](1.755,-0.25) rectangle (3.55,0.25);
		\end{tikzpicture}
		\caption{The lower bound instance in the proof of Theorem~\ref{thm: 2-IFS UW} for $n=4$. $f^*_{UW}$ represents the utilitarian welfare maximizing facility placement, whilst $f^*_{2IFS}$ maximizes utilitarian welfare under the constraints of 2-IFS. The red intervals denote locations that are infeasible under 2-IFS.}
				\label{fig:UW}
\end{figure}

\begin{theorem}\label{thm: 2-IFS UW}
The price of 2-IFS for utilitarian welfare is 2, and this bound is tight.
\end{theorem}
\begin{proof}[Lower Bound Proof]
Suppose $n$ is even, and that the agents are located at $\frac{1}{2n}-\epsilon$, $\frac{3}{2n}-2\epsilon$, $\dots$, $\frac{n-1}{2n}-\frac{n}{2}\epsilon$, $\frac{n+1}{2n}+\frac{n}{2}\epsilon$, $\dots$, $\frac{2n-3}{2n}+2\epsilon$, $\frac{2n-1}{2n}+\epsilon$ for some sufficiently small $\epsilon$ (see, e.g. Figure~\ref{fig:UW}). Under this symmetric profile, either a facility location of $0$ or $1$ leads to the maximum utilitarian welfare of $\frac{n}{2}$. The only facility locations satisfying 2-IFS are within the interval $[\frac12-\frac{n}{2}\epsilon, \frac12+\frac{n}{2}\epsilon]$. Any location in this interval gives the same utilitarian welfare as there are an equal number of agents on both sides, so suppose the facility is at $\frac{1}{2}$. This corresponds to a utilitarian welfare of $\frac{n}{4}+\epsilon n(1+\frac{n}{2})$. Taking the limit $\epsilon\rightarrow 0$ gives a ratio of $2$.
\end{proof}
\begin{theorem}\label{thm: 2-UFS UW}
The price of 2-UFS for utilitarian welfare is 2, and this bound is tight.
\end{theorem}
As the price of fairness for utilitarian welfare is the same for both proportional fairness axioms, it may be desirable to implement 2-UFS in favour of 2-IFS when loss of utilitarian welfare is the primary concern.

			\subsubsection{Egalitarian Welfare}
			The egalitarian welfare (EW) is an alternate measure of fairness frequently observed in the literature, focusing on the worst off agent. Our price of fairness analysis gives an insight into the tradeoff between egalitarian welfare/maximin fairness and proportional fairness in the OFLP.
			
Our first result is that the price of 2-IFS is $1$, meaning that a mechanism that maximizes egalitarian welfare is guaranteed to satisfy 2-IFS. This follows from Proposition~\ref{prop2IFS}, which states that a 2-IFS solution (in which every agent obtains at least $\frac{1}{2n}$ utility) always exists.
			
 			\begin{proposition}\label{prop:Po2IFSEW}
The price of 2-IFS for egalitarian welfare is $1$.
\end{proposition}
\begin{proof}
We know a $2-$IFS solution must always exist, meaning that under any agent location profile, there exists a facility location such that every agent is at least $\frac{1}{2n}$ distance from the facility. It follows immediately that a solution maximizes egalitarian welfare satisfies $2-$IFS. 
\end{proof}

On the other hand, we find that the price of 2-UFS is noticeably worse, taking a linear factor of $n-1$. The intuition behind this is that a coalition of $n-1$ agents at one point can ensure that the facility is distant from their location (and closer to the remaining agent's location) by a `factor' of $n-1$ (see, e.g. Figure~\ref{fig:EW}).
\begin{theorem}\label{thm: 2-UFS EW}
The price of 2-UFS for egalitarian welfare is $n-1$.
\end{theorem}
\begin{proof}
We first prove that the lower bound is $n-1$. It suffices to consider $n\geq 3$. Consider the location profile with $1$ agent at $\frac{1}{2n}-\epsilon$ and $n-1$ agents at $\frac{n+1}{2n}+\epsilon$ for sufficiently small $\epsilon>0$, (see, e.g. Figure~\ref{fig:EW}). The optimal solution places the facility at $1$ resulting in an egalitarian welfare of $\frac{n-1}{2n}-\epsilon$. The only 2-UFS solutions are in the interval $[\frac{1}{n}-\epsilon,\frac{1}{n}+\epsilon]$, and the solution of $\frac{1}{n}+\epsilon$ results in an egalitarian welfare of $\frac{1}{2n}+2\epsilon$. As $\epsilon\rightarrow 0$, the ratio approaches $n-1$.

We now prove that the upper bound is $n-1$. Firstly, it clearly suffices to consider location profiles where groups contain at most $n-1$ agents. Suppose there exists such an $x$ where $\min_i u(f^*_{EW}(x),x_i)\geq \frac{n-1}{2n}$, i.e. there is a solution where every agent has at least $\frac{n-1}{2n}$ utility. Then this also satisfies 2-UFS and results in an egalitarian ratio of $1$. Therefore the maximum ratio must have $\min_i u(f^*_{EW}(x),x_i)< \frac{n-1}{2n}$. Due to 2-UFS, we also have $\max_{y\in 2UFS(x)}\min_i u(y,x_i)\geq \frac{1}{2n}$. The theorem statement follows from dividing these two terms.
\end{proof}	
\begin{figure}
		\centering
		\begin{tikzpicture}
			\draw[] (-3.5,0) -- (3.5,0); 
			\draw [thick] (-3.5,0.3)  -- (-3.5,-0.3);  
			\draw [thick] (3.5,0.3)  -- (3.5,-0.3); 
		
			\node at (-3.5, -0.5) {$0$}; 
			\node at (3.5, -0.5) {$1$}; 
			\node at (0.6, -0.5) {$x_2\dots x_n$}; 
			\node at (-3, -0.5) {$x_1$}; 
		
			\filldraw[fill=red!40, rounded corners](-3.55,-0.25) rectangle (-2.45,0.25); 
			\node at (3.5, 0.73) {$f^*_{EW}$}; 
			\node at (-2.4,0.73) {$2UFS(x)$};
			\filldraw[fill=red!40, rounded corners](-2.35,-0.25) rectangle (3.55,0.25);
		\end{tikzpicture}
		\caption{The instance in the proof of Theorem~\ref{thm: 2-UFS EW}. $f^*_{EW}$ represents the egalitarian welfare maximizing facility placement, whilst $2UFS(x)$ represents the interval of facility placements satisfying 2-UFS. The red intervals denote locations that are infeasible under 2-UFS.}
		\label{fig:EW}
\end{figure}
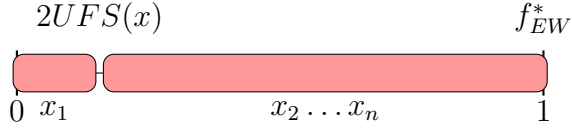

\subsection{Incompatibility with Strategyproofness}
In mechanism design, the normative property of \emph{strategyproofness} is often sought as it disincentivizes agents from misreporting their true location.

\begin{definition}[Strategyproofness]
A (deterministic) mechanism $f$ is strategyproof if for every agent $i\in N$, we have for every $x_i$, $x_i'$ and $\hat{x}_{-i}$, $$u(f(x_i,\hat{x}_{-i}),x_i)\geq u(f(x_i',\hat{x}_{-i}),x_i).$$
\end{definition}

We say that a randomized mechanism is \emph{strategyproof in expectation} if no agent can improve its expected utility by misreporting its own location.	
	
	We note that no strategyproof and deterministic mechanism can achieve any approximation of IFS (and therefore also UFS).

		\begin{proposition}\label{prop:detnofair}
			There exists no deterministic and strategyproof mechanism that achieves any approximation of IFS. 
			\end{proposition}
			\begin{proof}
				From the characterization by \citet{FeSu15a}, it can be seen that for any deterministic and strategyproof mechanism, there exists a location profile where the facility is placed at an agent's location. Such a mechanism does not satisfy any approximation of IFS.	
				\end{proof}	

Since strategyproofness is incompatible with our fairness axioms, we are interested in the performance of proportionally fair mechanisms in our model when accounting for agent strategic behaviour. Such performance can be quantified by the price of anarchy, and the price of stability.
\subsection{$\epsilon$-Price of Anarchy and $\epsilon$-Price of Stability}
In this section, we compute the loss of efficiency by agents misreporting their location (in a pure Nash equilibrium of reports) under the mechanisms $f^*_{2IFS}$ and $f^*_{2UFS}$. Recall these are the mechanisms which maximize utilitarian welfare under the constraints of 2-IFS and 2-UFS, respectively. This efficiency loss can be quantified in the `worst-case' sense, by the \emph{price of anarchy} \citep{KoPa99,NRTV07}, or in the `best case' sense, by the \emph{price of stability} \citep{ADK+08}.

However, for $f^*_{2IFS}$ and $f^*_{2UFS}$, we show that a pure Nash equilibrium may not necessarily exist, and hence the price of anarchy is not well-defined.

\begin{proposition}\label{prop: noNE}
A pure Nash equilibrium may not exist for $f^*_{2IFS}$ or $f^*_{2UFS}$.
\end{proposition}
\begin{proof}
For simplicity, we prove this statement for $f^*_{2IFS}$. The same arguments hold verbatim for $f^*_{2UFS}$.

We define a sufficiently small constant $\epsilon>0$, and consider the location profile $x=(\frac{1}{4}-\epsilon,\frac{3}{4}+\epsilon)$. We denote the reported location profile as $x'=(x'_1,x'_2)$, and prove this statement by considering cases on agent 1's reported location $x_1'$. Note that under a pure Nash equilibrium, $f^*_{2IFS}$ cannot place the facility in the interval $[0,\frac{1}{2}-\epsilon)$, as agent 1 can change its report to $x_1'=\frac{1}{4}-\epsilon$ to strictly increase its utility. Similarly, under a pure Nash equilibrium, $f^*_{2IFS}$ cannot place the facility in the interval $(\frac{1}{2}+\epsilon,1]$, as agent 2 can change its report to $x_2'=\frac{3}{4}+\epsilon$ to strictly increase its utility.

\textbf{Case 1 ($x_1'<\frac{1}{4}-\epsilon$)}: If $x_2'\leq \frac{3}{4}$, then $f^*_{2IFS}$ places the facility at $1$, and thus this is not a pure Nash equilibrium. If $x_2'>\frac{3}{4}$, then $f^*_{2IFS}$ places the facility at $x_1'+\frac{1}{4}<\frac{1}{2}-\epsilon$, thus this is also not a pure Nash equilibrium.

\textbf{Case 2 ($x_1'\geq \frac{1}{4}$)}: If $f^*_{2IFS}$ places the facility at a location strictly right of $0$, then this is not a pure Nash equilibrium as agent 2 can report $x_2'=1$ to move the facility to $0$, improving its utility. If $f^*_{2IFS}$ places the facility at $0$, then this is also not a pure Nash equilibrium as agent 1 can change its report to $x_1'=\frac{1}{4}-\epsilon$ to strictly increase its utility.

\textbf{Case 3 ($x_1'\in [\frac{1}{4}-\epsilon,\frac{1}{4})$)}: Recall that under a pure Nash equilibrium, $f^*_{2IFS}$ cannot place the facility in the interval $(\frac{1}{2}+\epsilon,1]$. Due to $x_1'$, $f^*_{2IFS}$ also cannot place the facility in the interval $[0,x_1'+\frac{1}{4})$. If $x_2'\leq \frac{3}{4}$, then $f^*_{2IFS}$ places the facility at $1$, and thus this is not a pure Nash equilibrium. Suppose that $x_2'>\frac{3}{4}$, meaning the facility must be placed in the interval $[x_1'+\frac{1}{4},x_2'-\frac{1}{4}]$. As $f^*_{2IFS}$ places the facility at the leftmost point of the optimal interval, it places the facility at $x_1'+\frac{1}{4}$. Thus, for any $x_1'\in [\frac{1}{4}-\epsilon,\frac{1}{4})$, there exists some sufficiently small $\delta>0$ such that agent 1 can instead report $x_1'+\delta$ to improve its utility, so by definition there does not exist a pure Nash equilibrium. In other words, agent 1 can continually shift its reported location asymptotically closer to $\frac{1}{4}$ to improve its utility, but from Case 2, we know that $x_1'$ cannot reach $\frac{1}{4}$ as otherwise there is no pure Nash equilibrium. 
\end{proof}

As a result, we turn to proving existence of the approximate notion of pure $\epsilon$-Nash equilibria, and computing the corresponding notions of $\epsilon$-price of anarchy and $\epsilon$-price of stability. 

\begin{definition}[\citet{TaVa07}]
A \emph{pure} $\epsilon$-\emph{Nash equilibrium} is a profile of reported agent locations $x'=(x_1',\dots,x_n')$ such that no single agent can improve its own utility (with respect to its true location) by strictly more than $\epsilon$ by changing its reported location. A pure Nash equilibrium is a pure $\epsilon$-Nash equilibrium where $\epsilon=0$.
\end{definition}

To prove the following theorems, we divide the space of agent location profiles into several subcases, and for each subcase, we describe a pure $\epsilon$-Nash equilibrium.
\begin{theorem}\label{thm: nasheq2IFS}
For any $\epsilon>0$, a pure $\epsilon$-Nash equilibrium always exists for $f^*_{2IFS}$.
\end{theorem}

\begin{theorem}\label{thm: nasheq2UFS}
For any $\epsilon>0$, a pure $\epsilon$-Nash equilibrium always exists for $f^*_{2UFS}$.
\end{theorem}

In real-world settings, the value of $\epsilon$ could represent a discretization of the domain, or the smallest distance of which an agent can change their reported location.

For a mechanism $f$, the $\epsilon$-price of anarchy (resp. stability) is defined as the worst-case ratio (over all location profiles $x$) between the utilitarian welfare corresponding to all agents reporting truthfully and the \emph{minimum} (resp. \emph{maximum}) utilitarian welfare corresponding to agents reporting in a pure $\epsilon$-Nash equilibrium.

\begin{definition}
Given $f$ and $x$, define the set of pure $\epsilon$-Nash equilibria location profiles as $\epsilon$-$Equil(f,x)$. The $\epsilon$-price of anarchy for utilitarian welfare is defined as:

$$\epsilon\text{-}PoA(f):=\max_{x\in X^n}\frac{\sum_i u(f(x),x_i)}{\min_{x'\in \epsilon \text{-}Equil(f,x)}\sum_i u(f(x'),x_i)}.$$
The $\epsilon$-price of stability for utilitarian welfare is defined as:
$$\epsilon\text{-}PoS(f):=\max_{x\in X^n}\frac{\sum_i u(f(x),x_i)}{\max_{x'\in \epsilon\text{-}Equil(f,x)}\sum_i u(f(x'),x_i)}.$$
\end{definition}

We now proceed to find $\epsilon$-price of anarchy bounds for utilitarian welfare. The same proof arguments can be applied to find identical bounds for both $f^*_{2IFS}$ and $f^*_{2UFS}$.

\begin{theorem}\label{thm: PoA2IFS}
For any $\epsilon\in (0,\frac{1}{n})$, the $\epsilon$-price of anarchy for $f^*_{2IFS}$ and $f^*_{2UFS}$ of utilitarian welfare is at least $\frac{2n-1+n\epsilon}{1-n\epsilon}$. The price of anarchy is unbounded for $\epsilon\geq \frac{1}{n}$.
\end{theorem}
\begin{proof}
Suppose for any $\epsilon\in (0,\frac{1}{n})$ that we have the (true) agent location profile $x=(\frac{1}{2n}-\frac{\epsilon}{2},\frac{1}{2n}-\frac{\epsilon}{2},\dots,\frac{1}{2n}-\frac{\epsilon}{2})$. We show that the location profile $x'=(1,\dots,1)$ is a pure $\epsilon$-Nash equilibrium for $x$. Under $x'$, both $f^*_{2IFS}$ and $f^*_{2UFS}$ place the facility at $0$, causing each agent to have $\frac{1}{2n}-\frac{\epsilon}{2}$ utility. An agent can only change the facility position by deviating to a reported location in $[0,\frac{1}{2n})$, causing the facility to instead be placed somewhere in $(0,\frac{1}{n})$. This results in the agent receiving a utility of $u(x_i)<\frac{1}{2n}+\frac{\epsilon}{2}$, which is an increase of at most $\epsilon$. Since no agent can improve its utility by greater than $\epsilon$ by misreporting, $x'$ is a pure $\epsilon$-Nash equilibrium. Now the utilitarian welfare under $x$ is $n-\frac{1}{2}+\frac{n\epsilon}{2}$, whilst under $x'$ it is $\frac{1}{2}-\frac{n\epsilon}{2}$ (w.r.t. $x$). Hence the $\epsilon$-price of anarchy is at least $\frac{2n-1+n\epsilon}{1-n\epsilon}$ for $\epsilon\in (0,\frac{1}{n})$.

For $\epsilon\geq\frac{1}{n}$, the (true) agent location profile $x=(0,\dots,0)$ has a corresponding pure $\epsilon$-Nash equilibrium of $x'=(1,\dots,1)$ which results in each agent having 0 utility. This can be seen as no agent can improve its utility by $\frac{1}{n}$ or greater, as an agent can only change the facility to a location in $(0,\frac{1}{n})$ by changing its report to a location in $(0,\frac{1}{2n})$. As each agent has 0 utility under the pure $\epsilon$-Nash equilibrium, the $\epsilon$-price of anarchy is unbounded for $\epsilon\geq \frac{1}{n}$.
\end{proof}
\begin{theorem}\label{thm: PoA2UFS}
For any $\epsilon\in (0,\frac{1}{2n})$, the $\epsilon$-price of anarchy for $f^*_{2IFS}$ and $f^*_{2UFS}$ of utilitarian welfare is at most $\frac{2n}{1-2n\epsilon}$.
\end{theorem}
\begin{proof}
Under a pure $\epsilon$-Nash equilibrium, each agent must have at least $\frac{1}{2n}-\epsilon$ utility.  This is because an agent can achieve at least $\frac{1}{2n}$ utility by reporting its true location. Therefore the utilitarian welfare under a pure $\epsilon$-Nash equilibrium must be at least $\frac{1}{2}-n\epsilon$. Now the utilitarian welfare under any instance is at most $n$, from all agents being located at $0$ and the facility being placed at $1$. The theorem statement follows from dividing these terms.
\end{proof}

As $\epsilon\rightarrow 0$, we see that the $\epsilon$-price of anarchy of $f^*_{2IFS}$ and $f^*_{2UFS}$ is linear and thus the mechanisms perform quite poorly in the worst-case equilibria. In contrast, we prove asymptotically tight constant bounds on the $\epsilon$-price of stability. To prove the lower bound, we give a location profile where the optimal $f^*_{2IFS}$ and $f^*_{2UFS}$ facility placements are near an interval endpoint, but the only $\epsilon$-Nash equilibria facility placement is in the middle of the agents.

\begin{theorem}\label{lem: PoSlower}
For $f^*_{2IFS}$ and $f^*_{2UFS}$, if $\epsilon\in (0,\frac{1}{2n})$, the $\epsilon$-price of stability is at least
$$\frac{4n^2-4n+4+8n\epsilon}{2n^2+n+2+(2n^3+4n^2-8n)\epsilon}.$$
This expression approaches $2$ as $\epsilon\rightarrow0$ and $n\rightarrow\infty$.
\end{theorem}
\begin{proof}
Suppose we have even $n$, and that the (true) agent location profile is $x=(0,\frac{3}{2n}-2\delta,\frac{5}{2n}-3\delta,\dots,\frac{n-1}{2n}-\frac{n}{2}\delta,\frac{n+1}{2n}+\frac{n}{2}\delta,\frac{n+3}{2n}+(\frac{n}{2}-1)\delta,\dots,\frac{2n-1}{2n}+\delta)$, where $\delta>\epsilon$ and $\delta-\epsilon$ is sufficiently small. Here, $f^*_{2IFS}$ and $f^*_{2UFS}$ place the facility at $\frac{1}{2n}$, which results in a utilitarian welfare of
\begin{align*}
\frac{1}{2n}+(\frac{1}{n}+\frac{2}{n}+\dots+\frac{n-1}{n})+\delta=\frac{n^2-n+1}{2n}+\delta.
\end{align*}
Now note that under a pure $\epsilon$-Nash equilibrium, the facility can only be placed in the interval $[\frac{1}{2}-\frac{n}{2}\delta,\frac{1}{2}+\frac{n}{2}\delta]$. If the facility is placed less than $\frac{1}{2n}-\epsilon$ distance of an agent's true location, that agent can report their true location to gain strictly more than $\epsilon$ utility. The facility also cannot be placed in the interval $[0,\frac{1}{n}-2\delta]$ as the agent at $x_1=0$ can change its report to $x_1'=\frac{1}{2n}-\delta$ to gain more than $\epsilon$ utility.

The utilitarian welfare corresponding to the equilibrium facility placement of $\frac{1}{2}-\frac{n}{2}\delta$ is
\begin{align*}
&\frac{1}{2}+(\frac{1}{2n}+\frac{3}{2n}+\dots+\frac{n-3}{2n})+(\frac{1}{2n}+\frac{3}{2n}+\dots+\frac{n-1}{2n})+(\frac{n}{2}-1)(\frac{n}{2}+2)\delta+\delta\\
&=\frac{2n^2+n+2}{8n}+(\frac{n^2}{4}+\frac{n}{2}-1)\delta.
\end{align*}
Dividing our welfares and taking the limit $\delta\rightarrow\epsilon$ gives us the $\epsilon$-price of stability lower bound of $$\frac{4n^2-4n+4+8n\epsilon}{2n^2+n+2+(2n^3+4n^2-8n)\epsilon}.$$ Setting the limit $\epsilon\rightarrow 0$, the expression becomes equal to $2-\frac{6n}{2n^2+n+2}$, and thus we see it is monotonic increasing for $n\geq 1$, approaching a value of $2$. 
\end{proof}

We next prove that the $\epsilon$-price of stability for $f^*_{2IFS}$ and $f^*_{2UFS}$ has an upper bound of 2. The proof iterates through the $\epsilon$-equilibria in each subcase of the proof of Theorem~\ref{thm: nasheq2IFS}, and constructs the location profile that maximizes the ratio between the utilitarian welfare when agents report truthfully, and the utilitarian welfare corresponding to the given $\epsilon$-equilibrium facility placement.
\begin{theorem}\label{thm:PoSupper}
For $f^*_{2IFS}$ and $f^*_{2UFS}$, taking the limit $\epsilon \rightarrow 0$, the $\epsilon$-price of stability is at most 2.
\end{theorem}

As we have shown, when maximizing the utilitarian welfare under 2-IFS or 2-UFS, the degradation of efficiency under a $\epsilon$-Nash equilibrium can range from a constant factor to a linear factor. To avoid a pessimistic outcome, we may wish to employ a randomized mechanism, achieving strategyproofness along with 2-IFS or 2-UFS in expectation. We give examples of such mechanisms in the upcoming section.
\section{Randomized Mechanisms}
By using randomized mechanisms, we can achieve a better price of fairness for 2-IFS and 2-UFS, and overcome the incompatibility with strategyproofness. We define a \emph{randomized mechanism} as a probability distribution over deterministic mechanisms, and an agent's utility as its expected distance from the facility.

In the randomized setting, the optimal approximation of IFS and UFS for which a solution always exists is $\alpha=2$, as seen by setting 1 agent at $\frac{1}{2}$. Our fairness axioms are adapted as follows:

\begin{definition}[$\alpha$-Individual Fair Share (IFS) in expectation]
				A mechanism $f$ satisfies \emph{$\alpha$-Individual Fair Share in expectation ($\alpha$-IFS in expectation)} if for any location profile $x$, $$\mathbb{E}[u(f(x), x_i)]\ge \frac{1}{\alpha n} \quad \forall i\in N.$$
				\end{definition}
\begin{definition}[$\alpha$-Unanimous Fair Share (UFS) in expectation]
			A mechanism $f$ satisfies \emph{$\alpha$-Unanimous Fair Share in expectation ($\alpha$-UFS in expectation)} if for any location profile $x$ and any set of agents $S$ at the same location, $$\mathbb{E}[u(f(x), x_i)]\ge \frac{|S|}{\alpha n} \quad \forall i\in S.$$
		\end{definition}
We first show that when maximizing welfare in the randomized setting, it suffices to consider mechanisms which can only place the facility at $0$ or $1$.

\begin{lemma}\label{lem: 01opt}
Consider an agent location profile $x$. For every 2-IFS/UFS randomized mechanism that gives strictly positive probability to a facility placement between $0$ and $1$, there exists a 2-IFS/UFS randomized mechanism that only gives positive support to a facility placement at $0$ or $1$ that leads to weakly higher expected utility for each agent.
\end{lemma}		
		
\subsection{Strategyproofness}
From Proposition~\ref{prop:detnofair}, we know that in the deterministic setting, strategyproofness is incompatible with our proportional fairness axioms. In the randomized setting, the space of mechanisms is much larger and hence we are able to overcome this impossibility.

We first consider \textbf{Mechanism 2} from \citep{CYZ13}. Denoting the numbers of agents located in $[0,1/2]$ and $(1/2,1]$ by $n_1$ and $n_2$ respectively, \textbf{Mechanism 2} places the facility at $0$ with probability $\alpha$ and at $1$ with probability $(1-\alpha)$, where $\alpha=\frac{2n_1n_2+n^2_2}{n_1^2+n_2^2+4n_1n_2}$. This mechanism is known to be group strategyproof (in expectation) and $\frac{3}{2}-$approximates the utilitarian welfare. We show that this mechanism satisfies 2-UFS (and therefore also 2-IFS).

\begin{theorem}\label{thm: mech2}
\textbf{Mechanism 2} satisfies 2-UFS in expectation.
\end{theorem}
		
		\subsection{Egalitarian Welfare}
We now provide some results on egalitarian welfare. Specifically, we give a randomized, strategyproof mechanism which maximizes egalitarian welfare subject to the constraints of 2-IFS and 2-UFS in expectation.

The \textbf{Randomized Egalitarian Welfare mechanism} places the facility at $1$ if all agents are in $[0,\frac{1}{2}]$, at $0$ if all agents are in $(\frac{1}{2},1]$, and at $0$ or $1$ with 0.5 probability otherwise.

By considering cases, it is easy to see that this mechanism is optimal and satisfies ideal normative properties.

\begin{proposition}\label{prop:randEWSP}
\textbf{Randomized Egalitarian Welfare mechanism} is strategyproof in expectation, optimal for egalitarian-welfare, and satisfies 2-UFS.
\end{proposition}
\begin{proof}
We first prove strategyproofness. If all agents are in $[0,\frac{1}{2}]$ or all agents are in $(\frac{1}{2},1]$, then each agent has at least $\frac{1}{2}$ expected utility. Any misreport either causes their expected utility to either stay the same or be reduced to $\frac{1}{2}$ from the facility being placed at $0$ or $1$ with probability $\frac{1}{2}$ each. If there is at least one agent in each interval, then an agent can only affect the outcome if it is the only agent in its interval and it misreports to be in the other interval, but this weakly reduces the agent's expected utility.

We now prove egalitarian welfare optimality and satisfaction of 2-UFS. The cases where all agents are in $[0,\frac{1}{2}]$ and all agents are in $(\frac{1}{2},1]$ are trivial, so it remains to examine the case where both intervals have at least one agent. An agent at $x_i$ has $\frac{1}{2}x_i+\frac{1}{2}(1-x_i)=1/2$ expected distance from the facility, hence this mechanism satisfies 2-UFS in expectation. By Lemma~\ref{lem: 01opt}, it suffices to only consider mechanisms which can only place the facility at 0 or 1. Suppose that instead of having $\frac{1}{2}$ probability of placing the facility at either endpoint, we place the facility at $1$ with $\frac{1}{2}+p$ probability and at $0$ with $\frac{1}{2}-p$ probability, where $p\in (0,\frac{1}{2}]$. The expected utility of the rightmost agent is $x_n(\frac{1}{2}-p)+(1-x_n)(\frac{1}{2}+p)=\frac{1}{2}+p(1-2x_n)<\frac{1}{2}$. By a symmetric argument, if the facility was placed at $1$ with $\frac{1}{2}-p$ probability and at $0$ with $\frac{1}{2}+p$ probability, the expected utility of the leftmost agent would be strictly less than $\frac{1}{2}$. Hence, our mechanism is optimal in this case.
\end{proof}
\begin{remark}
As each agent has at least $1/2$ expected distance from the facility under the \textbf{Randomized Egalitarian Welfare mechanism}, this mechanism even satisfies 1-IFS for $n\geq 2$.
\end{remark}
In other words, the approximation ratio of this mechanism for egalitarian welfare is $1$. Recall that the price of fairness can be interpreted as the approximation ratio of the respective optimal mechanism that satisfies the fairness constraint. This leads us to the following corollary.

\begin{corollary}\label{cor:PoF2UFSEW}
In the randomized setting, the price of fairness of 2-UFS for egalitarian welfare is $1$.
\end{corollary}
This is in stark contrast to the deterministic setting where the respective price of fairness is $n-1$.

		\subsection{2-IFS}
We now analyze utilitarian welfare, beginning with the axiom of 2-IFS. Consider the randomized mechanism below which maximizes the utilitarian welfare subject to 2-IFS:

\textbf{2-IFS Randomized mechanism}
\begin{itemize}
\item If $\sum_{i=1}^nx_i=\frac{n}{2}$, place the facility at $0$ with probability $\frac{1}{2}$ and at $1$ with probability $\frac{1}{2}$.
\item If $\sum_{i=1}^nx_i>\frac{n}{2}$,
\begin{itemize}
\item If $x_1\geq \frac{1}{2n}$, place the facility at $0$.
\item If $x_1<\frac{1}{2n}$, place the facility at $0$ with probability $1-\alpha$, and at $1$ with probability $\alpha$, where $\alpha=\frac{1-2nx_1}{2n(1-2x_1)}$.
\end{itemize}
\item If $\sum_{i=1}^nx_i<\frac{n}{2}$,
\begin{itemize}
\item If $x_n\leq 1-\frac{1}{2n}$, place the facility at $1$.
\item If $x_n>1-\frac{1}{2n}$, place the facility at $0$ with probability $1-\beta$, and at $1$ with probability $\beta$, where $\beta=\frac{1-2nx_n}{2n(1-2x_n)}$.
\end{itemize}
\end{itemize}
The intuition behind this mechanism is as follows. When $\sum_{i=1}^nx_i=\frac{n}{2}$, both facility locations of $0$ and $1$ are tied in terms of maximizing utilitarian welfare, and by placing the facility at either location with probability $\frac{1}{2}$, we achieve 2-IFS in expectation. When $\sum_{i=1}^nx_i>\frac{n}{2}$, the optimal facility location is $0$, so the mechanism places the facility there if it does not violate 2-IFS for any agent, else it places the facility at $1$ with the minimum probability that ensures 2-IFS is ensured for all agents. The case where  $\sum_{i=1}^nx_i<\frac{n}{2}$ is symmetric.

Our proof of the mechanism's welfare-optimality is based on its intuition.

\begin{lemma}\label{lem:2IFSRand}
\textbf{2-IFS Randomized mechanism} is optimal for utilitarian welfare amongst all randomized mechanisms satisfying 2-IFS in expectation.
\end{lemma}

We now prove, using an algebraic approach, a tight, constant approximation ratio for this mechanism.

\begin{theorem}\label{thm:2IFSRand}
\textbf{2-IFS Randomized mechanism} has an approximation ratio for utilitarian welfare of $\frac{12}{11}\approx 1.091$.
\end{theorem}

This implies the following price of fairness result for 2-IFS.

\begin{corollary}\label{cor:PoF2IFSUW}
In the randomized setting, the price of fairness of 2-IFS for UW is $\frac{12}{11}\approx 1.091$.
\end{corollary}

\subsection{2-UFS}
We now move to analyze the axiom of 2-UFS in the context of utilitarian welfare. As in the previous subsection, we begin by describing a randomized mechanism which maximizes the utilitarian welfare subject to the 2-UFS constraint:

\textbf{2-UFS Randomized mechanism}
\begin{itemize}
\item Order the $m$ \textbf{unique} agent locations such that $x_1<x_2<\dots<x_m$.
\item Let $S_1,\dots,S_m$ denote the groups of agents at the $m$ unique agent locations.
\item If $\sum_{i=1}^m |S_i|x_i=\frac{n}{2}$, place the facility at $0$ with probability $\frac{1}{2}$ and at $1$ with probability $\frac{1}{2}$.
\item If $\sum_{i=1}^m |S_i|x_i>\frac{n}{2}$,
\begin{itemize}
\item Let $k$ denote the index of the largest unique agent location satisfying $x_k<\frac{1}{2}$.
\item For $i$ in $\{1,\dots,k\}$, set $\alpha_i=\frac{|S_i|-2nx_i}{2n(1-2x_i)}$. 
\item Letting $\alpha=\max\{\alpha_1,\dots,\alpha_k\}$, place the facility at $0$ with probability $1-\alpha$ and at $1$ with probability $\alpha$.
\end{itemize}
\item If $\sum_{i=1}^m |S_i|x_i<\frac{n}{2}$,
\begin{itemize}
\item Let $k$ denote the index of the smallest unique agent location satisfying $x_k>\frac{1}{2}$.
\item For $i$ in $\{k,\dots,m\}$, set $\alpha_i=\frac{|S_i|-2nx_i}{2n(1-2x_i)}$.
\item Letting $\alpha=\min\{\alpha_k,\dots,\alpha_m\}$, place the facility at $0$ with probability $1-\alpha$ and at $1$ with probability $\alpha$.
\end{itemize}
\end{itemize}

This mechanism is similar to the 2-IFS Randomized mechanism, but we must now iterate through the groups of agents to find the optimal value of $\alpha$ that guarantees 2-UFS for all agents. Specifically, if $\sum_{i=1}^m |S_i|x_i>\frac{n}{2}$, then $\alpha_i$ denotes the smallest probability weight on location $1$ such that 2-UFS is achieved for $S_i$. Hence by setting $\alpha$ to be the largest $\alpha_i$, we achieve 2-UFS for all agents.

Again, our proof of this mechanism's optimality is based on the aforementioned intuition.

\begin{lemma}\label{lem:2UFSRand}
\textbf{2-UFS Randomized mechanism} is optimal for utilitarian welfare amongst all randomized mechanisms satisfying 2-UFS in expectation.
\end{lemma}

Surprisingly, imposing the stronger fairness axiom of 2-UFS as opposed to 2-IFS has a minimal effect on the welfare-optimal mechanism's approximation ratio. Again, the approximation ratio is computed algebraically.
\begin{theorem}\label{thm:2UFSRand}
\textbf{2-UFS Randomized mechanism} has an approximation ratio for utilitarian welfare of $\frac{2}{7}(1+2\sqrt{2})\approx 1.09384$.
\end{theorem}
From Theorem~\ref{thm:2UFSRand}, we have the following corollary.
\begin{corollary}\label{cor:PoF2UFSUW}
In the randomized setting, the price of fairness of 2-UFS for UW is $\frac{2}{7}(1+2\sqrt{2})\approx 1.09384$.
\end{corollary}

	\section{Extension 1: Proportional Fairness}\label{section:PF}
In our analyses of price of fairness and randomized mechanisms, we have considered 2-IFS and 2-UFS, which give minimum distance guarantees for individual agents and groups of agents at the same location, respectively. One downside of the 2-UFS definition is that agents located near each other but not at the same location are considered to be in separate groups. An axiom which accounts for groups of agents located relatively close to each other is Proportional Fairness (PF), from \citep{ALLW21}. As with IFS and UFS, a PF solution may not exist so we define approximate $\alpha-$PF as follows:		
			\begin{definition} [$\alpha$-PF]\label{def: alphaPF}
				Given a profile of locations $\boldsymbol{x}$,  a facility location $y$ satisfies \emph{$\alpha$-PF } if for any set of agents $S$ within range $r:=\max_{i\in S}\{x_i\}-\min_{i\in S}\{x_i\}$, $$u(y, x_i)\ge \frac{1}{\alpha}(|S|/(n))-r \quad \forall i\in S.$$
			\end{definition}

			Note that $\alpha-$PF implies $\alpha-$UFS, and therefore also implies $\alpha-$IFS.
				
				However, $\alpha-$UFS does not imply $\alpha-$PF, hence $\alpha-$PF is a stronger notion than $\alpha-$UFS.
				\begin{lemma}\label{lemma:UFPF-2}
					For $\alpha=2$, there exists an $\alpha-$UFS facility location $y$ that does not satisfy $\alpha-$PF.
				\end{lemma}
						\begin{proof}
					Assume that $n=10$ and that we have 2 groups of agents, with the first group of $7$ agents located at the point $c_1=0.35$, and the second group of $3$ agents located at the point $c_2=0.55$. We consider two intervals $B_1$ and $B_2$ respectively with centers $c_1$ and $c_2$ and radius $\frac{|S_1|}{2n}=\frac{7}{20}=0.35$ and $\frac{|S_2|}{2n}=\frac{3}{20}=0.15$. We set the facility location $y$ at the point $0.71$. Since point $y$ is outside of the two intervals, it satisfies 2-UFS. However, it does not satisfy the 2-PF inequality:
					$d(y,c_2)=0.16 \ngeq \frac{|S_1|}{20}+\frac{|S_2|}{20}-r=0.3$.			
				\end{proof}						
		We now show that a 2-PF solution always exists. The proof uses induction on the number of groups of agents at the same location.

				\begin{theorem}\label{exist2PF}
					A   2-PF solution always exists.
				\end{theorem}	
				\begin{proof}[Proof Sketch]					
					  We prove the theorem by induction on the number of groups $m$, where each group consists of agents at the same location. When $m=1$, i.e,  all the agents are at the same point, $2$-PF existence follows from $2$-UFS existence.   Assume that for any $k$ groups of agents where $k\leq m$, there exists a $2$-PF solution. Suppose we have $m+1$ groups of agents placed at centers $c_1, c_2,...,c_{m+1}$ which are ordered from left to right. Set an open interval $B_i$ with radius $\frac{|S_i|}{2n}$ around each center $c_i$. We consider several cases based on the intersection of open intervals.  If all the open intervals are disjoint, it can be shown there exists a point $y\in [0,1]$ which lies outside the union of open intervals $B_1\cup \dots \cup B_{m+1}$, satisfying the $2$-PF inequality. If there exists two open intervals, say $B_1$ and $B_2$, intersecting each other, they are merged with the agents placed at a new center $c'_1$. We then set an open interval $B'_1$ centered at $c'_1$ with radius $\frac{|S_1|}{2n}+\frac{|S_2|}{2n}$. Now, we have $m$ groups of agents placed at  $c'_1, c_3,\dots,c_{m+1}$, and from our inductive assumption, we know a $2$-PF solution exists.
				\end{proof}	
						
				From Theorem~\ref{optUFS}, we see that 2-PF is the optimal approximation of PF for the obnoxious facility location problem.

			\section{Extension 2: Hybrid Model}
In the hybrid model, agents either want to be located close to the facility (as in the FLP), or wish to be located far away from the facility (as in our OFLP model). Such a model has several real-world applications such as the placement of schools or religious places of worship; families with children or religious people would want to live near the facility for convenience, whilst others would want to be far from the facility due to the increased noise and traffic. In our model, we say an agent is type $C$ if it is a classic agent and prefers to be closer to the facility, and an agent is type $O$ if it is an obnoxious agent and prefers to be further away from the facility.\footnote{Our model is based on the model presented by \citet{FeSu15a}.} We denote the set of classic agents as $N_C$ and the set of obnoxious agents as $N_O$.

A type $C$ agent has utility $u(y,x_i)=1-d(y,x_i)$ and a type $O$ agent has utility $u(y,x_i)=d(y,x_i)$.\footnote{This choice of utility function is adapted from \citep{FeSu15a,ALLW21}. We refer the reader to those papers for a justification of the utility model.}

When defining IFS and UFS in the hybrid model, we use definitions consistent with \citep{ALLW21} and this paper. Our definition of Hybrid-Individual Fair Share (H-IFS) provides an appropriate distance guarantee for each agent.
\begin{definition}[Hybrid-Individual Fair Share (H-IFS)]
Given a profile of locations $x$, a facility location $y$ satisfies Hybrid-Individual Fair Share (H-IFS) if for all $i\in N_C$,
$$u(y,x_i)\geq \frac{1}{n}\quad \text{or, equivalently,}\quad d(y,x_i)\leq 1-\frac{1}{n},$$
and for all $i\in N_O$,
$$u(y,x_i)\geq \frac{1}{2n}\quad \text{or, equivalently,}\quad d(y,x_i)\geq \frac{1}{2n}.$$
\end{definition}

When defining UFS, we aim to capture proportional fairness guarantees for subsets of agents of the same type at the same location. Consider every subset $S\subseteq N$ of agents at the same location, where $S=S_C\cup S_O$. $S_C$ denotes the agents of $S$ that are of type $C$, and $S_O$ denotes the agents of $S$ that are of type $O$.

\begin{definition}[Hybrid-Unanimous Fair Share (H-UFS)]
Given a profile of locations $x$ such that a subset of $S_j\subseteq N$ agents\footnote{$j\in \{C,O\}$} share the same type and location, a facility location $y$ satisfies Hybrid-Unanimous Fair Share (H-UFS) if for all $i\in S_C$,
$$u(y,x_i)\geq \frac{|S_C|}{n} \quad \text{or, equivalently,}\quad d(y,x_i)\leq 1-\frac{|S_C|}{n},$$
and for all $i\in S_O$,
$$u(y,x_i)\geq \frac{|S_O|}{2n} \quad \text{or, equivalently,}\quad d(y,x_i)\geq \frac{|S_O|}{2n}.$$
\end{definition}
\begin{example}
Suppose there are $n-k$ type $C$ agents and $k$ type $O$ agents, all at the same location. The facility needs to be between $\frac{k}{2n}$ and $\frac{k}{n}$ distance from the group.
\end{example}
Although our definitions have a discrepancy in utility functions between the classic and obnoxious agents, we have specified them to be consistent with related literature and to be the optimal bounds such that a solution is guaranteed to exist. Furthermore, existence of a H-UFS solution under our definition implies existence of a solution under a weaker definition where a set $S_C$ of classic agents at the same location instead have a utility guarantee of $\frac{|S_C|}{2n}$.
\begin{theorem}\label{thm: hybrid}
Under the hybrid model, a H-UFS solution always exists.
\end{theorem}
				
\section{Discussion}
In this paper we have formulated proportional fairness axioms for the obnoxious facility location problem, and given welfare-optimal deterministic and randomized mechanisms satisfying these axioms. In both the deterministic and randomized setting, we prove tight price of fairness bounds for 2-IFS and 2-UFS, for the objectives of utilitarian and egalitarian welfare. These correspond to the approximation ratios of the respective welfare-optimal mechanisms. For the deterministic utilitarian welfare-optimal mechanisms, we prove existence of pure $\epsilon$-Nash equilibria, linear $\epsilon$-price of anarchy bounds, and constant $\epsilon$-price of stability bounds. We also give a randomized, strategyproof mechanism satisfying 2-UFS with a constant utilitarian approximation ratio.

Future directions to this work could stem from our proposed extension of 2-PF, as well as the extension of our proportional fairness axioms to the \emph{hybrid facility location model}. Further research could compute the price of fairness for these two extensions, and the prices of anarchy and stability for the welfare-optimal mechanisms in these settings.

Further extensions to the price of fairness results could involve different objective and utility functions. It is also worth analyzing the Nash equilibria of the randomized utilitarian welfare-optimal mechanisms, as they are not strategyproof in expectation. Although our proportional fairness axioms are incompatible with strategyproofness in the deterministic setting, we may consider weaker notions of strategyproofness which may be compatible with our fairness properties.

\section*{Acknowledgements}
Bo Li is funded by NSFC under Grant No. 62102333 and HKSAR RGC under Grant No. PolyU 15224823. We would like to acknowledge the helpful feedback and suggestions from Minming Li.

\bibliographystyle{named}
\bibliography{aamasbib}
\newpage
\appendix
\section{Proof of Theorem~\ref{thm: 2-IFS UW}}\label{app: 2-IFS USW}
\begin{customthm}{\ref{thm: 2-IFS UW}}
The price of $2-$IFS for utilitarian welfare is 2, and this bound is tight.
\end{customthm}
 			\begin{proof}
 			We first prove the lower bound. Suppose $n$ is even, and that the agents are located at $\frac{1}{2n}-\epsilon$, $\frac{3}{2n}-2\epsilon$, $\dots$, $\frac{n-1}{2n}-\frac{n}{2}\epsilon$, $\frac{n+1}{2n}+\frac{n}{2}\epsilon$, $\dots$, $\frac{2n-3}{2n}+2\epsilon$, $\frac{2n-1}{2n}+\epsilon$ for some sufficiently small $\epsilon$ (see, e.g. Figure~\ref{fig:UW}). Under this symmetric profile, either a facility location of $0$ or $1$ leads to the maximum utilitarian welfare of $\frac{n}{2}$. The only facility locations satisfying 2-IFS are within the interval $[\frac12-\frac{n}{2}\epsilon, \frac12+\frac{n}{2}\epsilon]$. Any location in this interval gives the same utilitarian welfare as there are an equal number of agents on both sides, so suppose the facility is at $\frac{1}{2}$. This corresponds to a utilitarian welfare of $\frac{n}{4}+\epsilon n(1+\frac{n}{2})$. Taking the limit $\epsilon\rightarrow 0$ gives a ratio of $2$.
 			
 			We now prove the upper bound. Suppose without loss of generality that the optimal facility location is 0. We define a sufficiently small $\epsilon>0$ which will be used in specifying certain agent locations, but is negligible in the computation of the welfare ratio.

\textbf{Case 1: Suppose that the optimal $2-$IFS facility location $y^*$ satisfies $y^*\in [x_k,x_{k+1}]$, where $k\leq \frac{n}{2}$.} 

Since $y^*$ is the optimal $2-$IFS facility location, any facility location left of $x_k$ must violate $2-$IFS. To see this, suppose there exists some $y'<x_k$ such that $y'$ satisfies $2-$IFS. A facility placed at $y'$ corresponds to a higher utilitarian welfare than $y^*$ as it is more distant from a majority of agents, leading to a contradiction. Furthermore, the welfare ratio 
\begin{align*}
&\frac{\sum_i u(f^*_{UW}(x),x_i)}{\max_{y\in 2IFS(x)}\sum_i u(y,x_i)}=\frac{\sum_{i=1}^nx_i}{\sum_{i=1}^k(y^*-x_i)+\sum_{i=k+1}^n(x_i-y^*)}
\end{align*}
increases with $\sum^k_{i=1}x_i$, so we must maximize $x_1,\dots,x_k$ whilst ensuring any facility location left of $x_k$ violates $2-$IFS. We therefore deduce that $x_i=\frac{2i-1}{2n}-i\epsilon$ for $i\in\{1,\dots,k\}$, and we therefore have
\begin{align*}
\max_{x\in [0,1]^n}\frac{\sum_i u(f^*_{UW}(x),x_i)}{\max_{y\in 2IFS(x)}\sum_i u(y,x_i)}&=\max_{x\in [0,1]^n}\frac{\sum_{i=1}^nx_i}{\sum_{i=1}^k(y^*-x_i)+\sum_{i=k+1}^n(x_i-y^*)}\\
&=\max_{x_{k+1},\dots,x_n\in [0,1]}\frac{\sum_{i=1}^k(\frac{2i-1}{2n}-i\epsilon)+\sum_{i=k+1}^n x_i}{\sum_{i=1}^k(y^*-(\frac{2i-1}{2n}-i\epsilon))+\sum_{i=k+1}^n(x_i-y^*)}.
\end{align*}
Now for the optimal facility location to be $0$, we must have $\sum_ix_i\geq \sum_i(1-x_i)$, or $\sum_ix_i\geq \frac{n}{2}$. We rewrite this as $\sum_{i=k+1}^nx_i\geq \frac{n}{2}-\sum_{i=1}^kx_i$. Now the welfare ratio increases as $\sum_{i=k+1}^nx_i$ decreases, so it is maximized w.r.t $x_{k+1},\dots,x_n$ when we have $\sum_{i=k+1}^nx_i= \frac{n}{2}-\sum_{i=1}^kx_i$ (which results in location $1$ being tied with $0$ as the optimal facility location).

Substituting this into the welfare ratio, we have
\begin{align*}
\max_{x\in [0,1]^n}\frac{\sum_i u(f^*_{UW}(x),x_i)}{\max_{y\in 2IFS(x)}\sum_i u(y,x_i)}&=\max_{x_{k+1},\dots,x_n\in [0,1]}\frac{\sum_{i=1}^k(\frac{2i-1}{2n}-i\epsilon)+\sum_{i=k+1}^n x_i}{\sum_{i=1}^k(y^*-(\frac{2i-1}{2n}-i\epsilon))+\sum_{i=k+1}^n(x_i-y^*)}\\
&=\frac{\sum_{i=1}^k(\frac{2i-1}{2n}-i\epsilon)+\frac{n}{2}-\sum_{i=1}^k(\frac{2i-1}{2n}-i\epsilon)}{(2k-n)y^*-\sum_{i=1}^k(\frac{2i-1}{2n}-i\epsilon)+\frac{n}{2}-\sum_{i=1}^k(\frac{2i-1}{2n}-i\epsilon)}\\
&=\frac{n/2}{(2k-n)y^*-2(\frac{1}{2n}+\dots+\frac{2k-1}{2n})+\frac{n}{2}+2\sum^k_{i=1}i\epsilon}\\
&=\frac{n/2}{(2k-n)y^*+\frac{n}{2}-\frac{k^2}{n}+2\sum^k_{i=1}i\epsilon}.
\end{align*}
Since $k\leq \frac{n}{2}$, shifting a facility within $(x_k,x_{k+1})$ slightly to the left causes the total utility to weakly increase as there are a greater number of agents who gain utility than those who lose utility. Therefore as $y^*$ is the optimal $2-$IFS facility, it must be as close to $x_k$ as possible, at $y^*=\frac{k}{n}-k\epsilon$. Substituting this into the welfare ratio (and ignoring the negligible $\epsilon$), we have
\begin{align*}
\max_{x\in [0,1]^n}\frac{\sum_i u(f^*_{UW}(x),x_i)}{\max_{y\in 2IFS(x)}\sum_i u(y,x_i)}&=\frac{n/2}{(2k-n)y^*+\frac{n}{2}-\frac{k^2}{n}}\\
&=\frac{n/2}{(2k-n)\frac{k}{n}+\frac{n}{2}-\frac{k^2}{n}}\\
&=\frac{n/2}{\frac{k^2}{n}-k+\frac{n}{2}}.
\end{align*}

Simple calculus shows that the denominator is minimized (w.r.t. $k\in (0,\frac{n}{2}]$) when $k=\frac{n}{2}$, resulting in a welfare ratio of $2$. Therefore when the optimal $2-$IFS facility location $y^*$ satisfies $y^*\in [x_k,x_{k+1}]$, where $k\leq \frac{n}{2}$, the price of $2-$IFS for utilitarian welfare is at most $2$.

\textbf{Case 2: Suppose that the \textit{unique} optimal $2-$IFS facility location $y$ satisfies $y^*\in [x_k,x_{k+1}]$, where $k> \frac{n}{2}$.}\footnote{We can disregard $k=n$ as we would have $y^*=1$, and as $0$ is the optimal facility location, we have $\sum_{i=1}^nx_i\ge \frac{n}{2}$, which corresponds to a welfare ratio under $2$.} We can assume uniqueness without loss of generality as a differing facility location $y^\dagger$ with the same utilitarian welfare as $y^*$ must satisfy $y^\dagger\in [x_j,x_{j+1}]$, where $j\leq \frac{n}{2}$. Similar to the previous case, any facility location right of $x_{k+1}$ must violate $2-$IFS as $y^*$ is the optimal facility location, and we have $y^*=x_{k+1}-\frac{1}{2n}$ as it lies to the right of the majority of agents, so shifting it leftwards would decrease the utilitarian welfare. We also remark that $x_k\le x_{k+1}-\frac{1}{n}$ as $y^*$ satisfies $2-$IFS.

We apply a sequence of transformations to the location profile where each transformation increases the welfare ratio. The transformations convert the location profile into an instance of Case 1 (where the optimal $2-$IFS facility location $y^*$ satisfies $y^*\in [x_k,x_{k+1}]$, where $k\leq \frac{n}{2}$).

The transformations are as follows:
\begin{itemize}
\item If $x_{k-1}$ (and $x_k$) are at $y^*-\frac{1}{2n}(=x_{k+1}-\frac{1}{n})$, shift $x_k$ rightwards to $x_k'=y^*+(y^*-x_k)$, causing the optimal $2-$IFS facility location $y^*$ to remain at the same location and/or satisfy Case 1.
\item If $x_{k-1}\in (x_{k+1}-\frac{2}{n},x_{k+1}-\frac{1}{n})$, shift $x_k$ rightwards to $x_k'=x_{k-1}+\frac{1}{n}$, causing the optimal $2-$IFS facility location $y^*$ to move leftwards to $y'=x_k'-\frac{1}{2n}(=x_{k-1}+\frac{1}{2n})$ and/or satisfy Case 1.
\item If $x_{k-1}\le x_{k+1}-\frac{2}{n}$, shift $x_k$ rightwards to $x_k'=x_{k+1}-\frac{1}{n}+\epsilon$, causing the optimal $2-$IFS facility location $y^*$ to move leftwards to $y'=x_k'-\frac{1}{2n}(=x_{k+1}-\frac{3}{2n}+\epsilon)$ and/or satisfy Case 1.
\end{itemize}
To justify the effect on $y^*$ from shifting $x_k$, recall that if $y^*$ still satisfies Case 2 after the shift, then it is still the rightmost location satisfying $2-$IFS as the majority of agents lie left of $y^*$. Now if $y^*$ changes to a location satisfying Case 1, then by our previous analysis the welfare ratio is at most 2, so we can disregard this scenario. Suppose that the first dot point occurs i.e. $y^*$ remains at the same location. Recall that the welfare ratio is 
$$\frac{\sum_i u(f^*_{UW}(x),x_i)}{\max_{y\in 2IFS(x)}\sum_i u(y,x_i)}=\frac{\sum_{i=1}^nx_i}{\sum_{i=1}^n|y^*-x_i|}.$$
The optimal facility location is still $0$ as the sum of agent locations increases, so the numerator strictly increases. The denominator remains the same as $x_k$ has moved to an equidistant location on the other side of $y^*$, and all other agents are at the same location. Thus this transformation increases the welfare ratio.

Consider the second and third dot points, i.e., $y^*$ moves to $y'=x_k'-\frac{1}{2n}$. Clearly, the numerator increases, so we now show that the denominator decreases. The change in utilitarian welfare from the transformation is 
\begin{align*}
\left[\sum_{i=1}^{k-1}\right. & \left.(y'-x_i) +(x_k'-y') + \sum_{i=k+1}^n(x_i-y')\right]-\left[\sum_{i=1}^k(y^*-x_i)+\sum_{i=k+1}^n(x_i-y^*)\right]\\
=&(y'-y^*)(2k-n-1)+(x_k'+x_k)-(y^*+y')<0.
\end{align*}
We know that $(y'-y^*)(2k-n-1)<0$ as $y'<y^*$ and $k\geq \frac{n+1}{2}$. We also know that $x_k\leq y^*-\frac{1}{2n}$ as $y^*$ satisfies $2-$IFS and that $y'=x_k'-\frac{1}{2n}$, so from this we deduce that $x_k'+x_k<y^*+y'$. Therefore the denominator of the welfare ratio decreases, and the transformation increases the welfare ratio.

As the transformations only require that $k>\frac{n}{2}$, we can repeatedly apply them (and update $x_k$ to be the rightmost agent left of $y^*$) until we have an instance of Case 1, which has been shown to have a maximum welfare ratio of $2$. Therefore the price of $2-$IFS for UW is at most $2$, and combining this with the lower bound of $2$ gives us the theorem statement.
\end{proof}
\section{Proof of Theorem~\ref{thm: 2-UFS UW}}
\begin{customthm}{\ref{thm: 2-UFS UW}}
The price of $2-$UFS for utilitarian welfare is 2, and this bound is tight.
\end{customthm}
\begin{proof}
The lower bound example in the proof of Theorem~\ref{thm: 2-IFS UW} implies that the price of $2-$UFS for UW is also at least $2$. It remains to prove the upper bound.

Similar to the proof of Theorem~\ref{thm: 2-IFS UW}, we suppose without loss of generality that the optimal facility location is $0$ and define a sufficiently small $\epsilon>0$. We also divide the proof into two cases:

\textbf{Case 1: Suppose that the optimal $2-$UFS facility location $y^*$ satisfies $y^*\in [x_k,x_{k+1}]$, where $k\leq \frac{n}{2}$.} To avoid contradicting the $2-$UFS optimality of $y^*$, the agent locations $x_1,\dots,x_k$ must be arranged such that any location left of $x_k$ violates $2-$UFS. Furthermore, those agents must be located such that $\sum_{i=1}^kx_i$ is maximized, as the welfare ratio increases with $\sum_{i=1}^kx_i$. We claim that this occurs when all $k$ agents are at the same location of $\frac{k}{2n}-\epsilon$. To see this, suppose that among agents $\{1,\dots,k\}$ that there are $m\leq k$ unique agent locations, and construct an open interval at each unique agent location with radius $\frac{c}{2n}$, where $c$ is the number of agents at the location. Any location within an open interval fails to satisfy $2-$UFS, so to maximize the welfare ratio and avoid a contradiction, the leftmost open interval must include $0$, and the $m$ intervals should overlap by an $\epsilon$ distance (to prevent a $2-$UFS facility being placed at the boundary between two intervals). An example of this with $m=k$ is for $i\in \{1,\dots,k\}$, $x_i=\frac{2i-1}{2n}-i\epsilon$. Therefore we see that for $m\leq k$ unique agent locations, the sum of agent locations is $\sum_{i=1}^kx_i=\frac{1}{2n}+\dots+\frac{2k-1}{2n}-m\epsilon$, which is maximized when all $k$ agents are at $\frac{k}{2n}-\epsilon$.

As in the $2-$IFS proof, we require $\sum_{i=1}^nx_i\geq \frac{n}{2}$ for the optimal facility location to be $0$, and since the welfare ratio decreases with $x_{k+1}+\dots+x_n$, we must have $\sum_{i=1}^nx_i=\frac{n}{2}$ and $x_{k+1}+\dots+x_n=\frac{n}{2}-(x_1+\dots+x_k)$. Furthermore, as $y$ is the optimal $2-$UFS location, it must take the leftmost $2-$UFS point within $[x_k,x_{k+1}]$, which is at $\frac{k}{n}-\epsilon$. Substituting these expressions into the welfare ratio gives
\begin{align*}
\frac{\sum_i u(f^*_{UW}(x),x_i)}{\max_{y\in 2UFS(x)}\sum_i u(y,x_i)}&=\frac{\sum_{i=1}^nx_i}{\sum_{i=1}^n|y^*-x_i|}\\
&=\frac{n/2}{\frac{k^2}{n}-k+\frac{n}{2}},
\end{align*}
which has been shown in the proof of Theorem~\ref{thm: 2-IFS UW} to attain a maximum of $2$. Therefore in this case, the price of $2-$UFS for utilitarian welfare is at most $2$.

\textbf{Case 2: Suppose that the \emph{unique} optimal $2-$UFS facility location $y^*$ satisfies $y^*\in [x_k,x_{k+1}]$, where $k> \frac{n}{2}$.} We can assume uniqueness without loss of generality as a differing facility location $y^\dagger$ with the same utilitarian welfare as $y^*$ must satisfy $y^\dagger\in [x_j,x_{j+1}]$, where $j\leq \frac{n}{2}$. We will apply a sequence of transformations which weakly increase the welfare ratio and result in a location profile satisfying Case 1. The transformation works as follows: shift $x_k$ rightwards to $x_k'=y^*+\frac{1}{2n}$. If there is already an agent at $y+\frac{1}{2n}$, then instead shift $x_k$ rightwards to $x_k'=y^*+\frac{1}{2n}+\epsilon_k$ where $\epsilon_k>0$ is sufficiently small, such that there are no other agents at $x_k'$.\footnote{Such a location always exists unless $x_n=1$, $y^*=1-\frac{1}{2n}$ and $x_1,\dots,x_{n-1}<y^*$, but it can easily be shown using $\sum_{i=1}^n\ge \frac{n}{2}$ that such a location profile corresponds to a welfare ratio less than 2.} This causes $y^*$ to remain at the same location and/or satisfy Case 1, as if $y^*$ still satisfies Case 2, it is the rightmost location satisfying $2-$UFS\footnote{The majority of agents are left of $y^*$, so shifting $y^*$ leftwards decreases the utilitarian welfare.}, and the location $y^*$ still satisfies $2-$UFS. Furthermore, the optimal facility location remains as $0$ as the sum of agent locations strictly increases. If $y^*$ satisfies Case 1 then we know the welfare ratio is at most $2$ so we disregard this scenario. We now show that otherwise the transformation increases the welfare ratio. Recall that the welfare ratio is
$$\frac{\sum_i u(f^*_{UW}(x),x_i)}{\max_{y\in 2UFS(x)}\sum_i u(y,x_i)}=\frac{\sum_{i=1}^nx_i}{\sum_{i=1}^n|y^*-x_i|}.$$
We know that before the shift, $y^*\geq x_k+\frac{1}{2n}$, so the numerator increases by at least $\frac{1}{n}$. The denominator either decreases, remains the same, or increases by at most $\epsilon_k$. Since $\epsilon_k$ is chosen to be sufficiently small, we conclude that this transformation causes the welfare ratio to weakly increase. By repeatedly applying these transformations and updating $x_k$ to be the rightmost agent left of $y^*$, we eventually arrive at a location profile satisfying Case 1, which we know has a welfare ratio of at most $2$. Hence the price of $2-$UFS for utilitarian welfare is at most $2$. As the price of $2-$UFS for utilitarian welfare is at least 2, the theorem statement follows. 
\end{proof}

\section{Proof of Theorem~\ref{thm: nasheq2IFS}}
\begin{customthm}{\ref{thm: nasheq2IFS}}
For any $\epsilon>0$, a pure $\epsilon$-Nash equilibrium always exists for $f^*_{2IFS}$.
\end{customthm}
\begin{proof}
Consider an arbitrary true agent location profile $x=(x_1,\dots,x_n)$ and sufficiently small $\epsilon>0$. Note that a pure $\epsilon$-Nash equilibrium is also a pure $\delta$-Nash equilibrium, where $\delta>\epsilon$.

\textbf{Case 1 ($n$ is even)}:

\textbf{Subcase 1a:} We first show that if $x_i\geq \frac{2i-1}{2n}$ for any $i\in \{1,\dots,\frac{n}{2}\}$, then a pure $\epsilon$-Nash equilibrium exists. Suppose this is the case, and let $j=\arg\min_{i\in [\frac{n}{2}]} \{x_i\geq \frac{2i-1}{2n}\}$. If $j=1$ (i.e. $x_1,\dots,x_n \geq \frac{1}{2n}$), then the reported location profile $(1,\dots,1)$ is a pure $\epsilon$-Nash equilibrium. This is because an agent can only influence the facility position by changing its report to a location in $[0,\frac{1}{2n})$, moving the facility from $0$ to a point in $(0,\frac{1}{n})$ and reducing its utility.

If $j>1$, then we will show that the reported location profile $x'=(x_1',\dots,x_n')=(\frac{1}{2n}-\epsilon,\dots,\frac{2(j-1)-1}{2n}-(j-1)\epsilon,1,\dots,1)$ is a pure $\epsilon$-Nash equilibrium. Under $x'$, the facility cannot be placed in $[0,\frac{j-1}{n}-(j-1)\epsilon)$ due to $x_1',\dots,x_{j-1}'$, and hence it is placed at $\frac{j-1}{n}-(j-1)\epsilon$. 

Suppose that for some agent $i\in \{1,\dots,j-1\}$ changes its report to some $x_i'\neq \frac{2i-1}{2n}-i\epsilon$. Under the resulting location profile, the facility moves to a location in $[0,\frac{1}{n}-2\epsilon]$ if $i=1$ and $[\frac{i-1}{n}-(i-1)\epsilon,\frac{i}{n}-(i+1)\epsilon]$ if $i\in \{2,\dots, j-2\}$, reducing the agent's utility. As agent $j-2$ is located at $\frac{2(j-2)-1}{2n}-(j-2)\epsilon$, agent $j-1$ (who is located at $x'_{j-1}=\frac{2(j-1)-1}{2n}-(j-1)\epsilon$) can improve its utility by reporting a new location $x''_{j-1}=x'_{j-1}+\epsilon_1$, where $\epsilon_1<\epsilon$. This causes the facility to shift to the right. However, agent $j-1$ cannot improve its utility by more than $\epsilon$, as if it reports a location $x''_{j-1}\geq x'_{j-1}+\epsilon$, the facility will be placed at $\frac{j-2}{2n}-(j-2)\epsilon$, reducing its utility. Hence agents $1,\dots,j-1$ cannot improve their utility by more than $\epsilon$ by changing their reported location.

Now consider agents $j,\dots,n$ whose true locations satisfy $x_j,\dots,x_n\geq \frac{2j-1}{2n}$ and have reported locations $x_j',\dots,x_n'=1$. As at least half of the agents must lie to the right of the facility, the facility takes the leftmost location satisfying $2-$IFS, even after any agent changes its report.\footnote{Recall that $f^*_{2IFS}$ selects the leftmost optimal location if there is a tie.} Hence an agent from $\{j,\dots,n\}$ can only influence the facility location by changing its report to a location in $(\frac{2(j-1)-1}{2n}-(j-1)\epsilon,\frac{2(j-1)+1}{2n}-(j-1)\epsilon)$, causing the facility to move to a location in $(\frac{j-1}{n}-(j-1)\epsilon,\frac{j+1}{n}-(j-1)\epsilon)$. It is easy to see that this strictly reduces the agent's utility. We have shown that no deviation by a single agent can cause its utility to increase by more than $\epsilon$, and hence $x'$ is a pure $\epsilon$-Nash equilibrium. Therefore a pure $\epsilon$-Nash equilibrium exists if $x_i\geq \frac{2i-1}{2n}$ for any $i\in \{1,\dots,\frac{n}{2}\}$.

\textbf{Subcase 1b:} By symmetry, we see that a $\epsilon$-Nash equilibrium also exists if $x_i\leq \frac{2i-1}{2n}$ for any $i\in \{\frac{n}{2}+2,\dots,n\}$. However, the exact symmetric argument does not work for $i=\frac{n}{2}+1$ if $x_{\frac{n}{2}+1}>\frac{n+3}{4n}$ as under the reported location profile $(0,\dots,0, \frac{n+3}{2n}+(\frac{n}{2}-1)\epsilon, \dots, 1-\frac{1}{2n}+\epsilon)$, agent $\frac{n}{2}+1$ can change its report from $x'_{\frac{n}{2}+1}=0$ to $\frac{n+1}{2n}+(\frac{n}{2})\epsilon$, causing the facility to move from $\frac{n+2}{2n}+(\frac{n}{2}-1)\epsilon$ to $\frac{1}{2n}$. This is because $f^*_{2IFS}$ breaks ties in favour of the leftmost optimal location, and results in increased utility as $x_{\frac{n}{2}+1}\in (\frac{n+3}{4n}, \frac{n+1}{2n}]$.

We now show that if $x_i>\frac{2i-1}{2n}$ for all $i\in \{\frac{n}{2}+2,\dots,n\}$ and $x_{\frac{n}{2}+1}\in (\frac{n+3}{4n},\frac{n+1}{2n}]$, a pure $\epsilon$-Nash equilibrium exists. It suffices to assume that $x_i<\frac{2i-1}{2n}$ for all $i\in \{1,\dots,\frac{n}{2}\}$, as otherwise we know from Subcase 1a that a pure $\epsilon$-Nash equilibrium exists. We divide this proof to two further subcases depending on where $x_{\frac{n}{2}+1}$ is located.

\textbf{Subcase 1bi:} In this subcase we consider $x_{\frac{n}{2}+1}<\frac{1}{2}+\frac{1}{2n}$. We claim that $x'=(\frac{1}{2n}-\epsilon,\frac{3}{2n}-2\epsilon,\dots,\frac{n-1}{2n}-(\frac{n}{2})\epsilon,0,\frac{n+3}{2n}+(\frac{n}{2}-1)\epsilon,\dots,1-\frac{1}{2n}+\epsilon)$ is a pure $\epsilon$-Nash equilibrium. Under $x'$, the facility is placed at the rightmost point of the only feasible interval, at $y'=\frac{n+2}{2n}+(\frac{n}{2}-1)\epsilon$, as there is a majority of agents left of the point. If any agent $i\in\{1,\dots,\frac{n}{2}-1\}$ changes its report, the facility will move to a location in $[\frac{1}{2n},\frac{1}{n}-2\epsilon]$ if $i=1$, and in $[\frac{i-1}{n}-(i-1)\epsilon,\frac{i}{n}-(i+1)\epsilon]$ if $i\in \{2,\dots, \frac{n}{2}-1\}$. This reduces agent $i$'s utility as $x_i<\frac{2i-1}{2n}$. We now consider agent $\frac{n}{2}$ (reporting $x'_{\frac{n}{2}}=\frac{n-1}{2n}-(\frac{n}{2})\epsilon$). Any location right of $y'$ is not feasible and under $x'$, there are $\frac{n}{2}+2$ agent reports, including $x'_{\frac{n}{2}}$, left of the facility. Hence agent $\frac{n}{2}$ can only influence the facility location by changing its report to a point in $(\frac{n+1}{2n}+(\frac{n}{2}-1)\epsilon,1]$, causing the facility to move to the leftmost point of the feasible interval, at $\frac{n-2}{2n}-(\frac{n}{2}-1)\epsilon$. This reduces its utility as $x_{\frac{n}{2}}<\frac{n-1}{2n}$. Next we consider agent $\frac{n}{2}+1$ (reporting $x'_{\frac{n}{2}+1}=0$), whose report can only change the facility's location within the feasible interval $[\frac{1}{2}-(\frac{n}{2})\epsilon,\frac{n+2}{2n}+(\frac{n}{2}-1)\epsilon]$. As we have $x_{\frac{n}{2}+1}<\frac{1}{2}+\frac{1}{2n}$, it is most optimal to have the facility at $y'=\frac{n+2}{2n}+(\frac{n}{2}-1)\epsilon$, achieved by reporting $x_{\frac{n}{2}+1}'=0$. Finally, we consider agents $\frac{n}{2}+2,\dots,n$. Similar to Subcase 1a, agent $\frac{n}{2}+2$ can improve its utility by strictly less than $\epsilon$ by reporting a location $x''_{\frac{n}{2}+2}=x'_{\frac{n}{2}+2}-\epsilon_1$, where $\epsilon_1<\epsilon$. If $\epsilon_1 \geq \epsilon$, the facility is placed at $\frac{n+4}{2n}+(\frac{n}{2}-2)\epsilon$, reducing the agent's utility. An agent $i\in \{\frac{n}{2}+3,\dots,n\}$ can only cause the facility to be in $[\frac{i-1}{n}+(i-1)\epsilon,\frac{i+1}{n}+(i-2)\epsilon]$. This results in a reduction of utility as $x_i>\frac{2i-1}{2n}$ for all $i\in \{\frac{n}{2}+2,\dots,n\}$. As no single agent can improve its utility by more than $\epsilon$ by changing its report, $x'$ is a pure $\epsilon$-Nash equilibrium and hence one exists for this subcase.

\textbf{Subcase 1bii:} In this subcase we consider $x_{\frac{n}{2}+1}\geq \frac{1}{2}+\frac{1}{2n}$. We claim that $x'=(\frac{1}{2n}-\epsilon,\frac{3}{2n}-2\epsilon,\dots,\frac{n-1}{2n}-(\frac{n}{2})\epsilon,\frac{n+1}{2n}+(\frac{n}{2})\epsilon,\frac{n+3}{2n}+(\frac{n}{2}-1)\epsilon,\dots,1-\frac{1}{2n}+\epsilon)$ is a pure Nash $\epsilon$-equilibrium. Here the facility is placed at the leftmost point of the optimal interval, at $\frac{1}{2}-(\frac{n}{2})\epsilon$. By using identical reasoning as in Subcase 1bi, it is easy to see that agents $1,\dots,\frac{n}{2}$, and agents $\frac{n}{2}+2,\dots,n$ cannot improve their utility by more than $\epsilon$ by misreporting. In this subcase, it is agent $\frac{n}{2}$ rather than agent $\frac{n}{2}+2$ who can improve its utility by less than $\epsilon$ by misreporting slightly to the right. Also, as in Subcase 1bi, agent $\frac{n}{2}+1$ can only change the facility's location to within the feasible interval $[\frac{1}{2}-(\frac{n}{2})\epsilon,\frac{n+2}{2n}+(\frac{n}{2}-1)\epsilon]$, but since we have $x_{\frac{n}{2}+1}\geq \frac{1}{2}+\frac{1}{2n}$, their optimal facility location of $\frac{1}{2}-(\frac{n}{2})\epsilon$ is achieved by their report of $x'_{\frac{n}{2}+1}=\frac{n+1}{2n}+(\frac{n}{2})\epsilon$.

\textbf{Subcase 1c:} It reminds to consider the subcase where $x_i< \frac{2i-1}{2n}$ for all $i\in \{1,\dots,\frac{n}{2}\}$ and $x_i> \frac{2i-1}{2n}$ for all $i\in \{\frac{n}{2}+1,\dots,n\}$. Under this subcase, we claim that the location profile $x'=(x_1',\dots,x_n')=(\frac{1}{2n}-\epsilon,\frac{3}{2n}-2\epsilon,\dots,\frac{n-1}{2n}-(\frac{n}{2})\epsilon,1,\dots,1)$ is a pure $\epsilon$-Nash equilibrium. Here, the facility is placed at $\frac{1}{2}-(\frac{n}{2})\epsilon$. By using the same arguments as in the first subcase, we see that agents $1,\dots,\frac{n}{2}-1$ who have reported locations $x_1'=\frac{1}{2n}-\epsilon, x_2'=\frac{3}{2n}-2\epsilon,\dots,x_{\frac{n}{2}-1}'=\frac{n-3}{2n}-(\frac{n}{2}-1)\epsilon$ have no incentive to change their report. Agent $\frac{n}{2}$ who reports location $x_{\frac{n}{2}}'=\frac{n-1}{2n}-(\frac{n}{2})\epsilon$ can improve its utility by strictly less than $\epsilon$, by changing its report to $x_{\frac{n}{2}}''=x_{\frac{n}{2}}'+\epsilon_1$, where $\epsilon_1<\epsilon$. Similar to Subcase 1a, if $\epsilon_1\geq \epsilon$, then the facility is placed at $\frac{n-2}{2n}-(\frac{n}{2}-1)\epsilon$, reducing the agent's utility. The agents $\frac{n}{2}+1,\dots,n$ who have reported their location as $1$ under $x'$ can only move the facility to a location in $(\frac{1}{2}-(\frac{n}{2})\epsilon,\frac{1}{2}+\frac{1}{n}-(\frac{n}{2})\epsilon)$ by changing its report to a location in $(\frac{n-1}{2n}-(\frac{n}{2})\epsilon,\frac{n+1}{2n}-(\frac{n}{2})\epsilon)$.\footnote{We note that $\sum_i x_i'<\frac{n}{2}-1+\frac{n}{4}(\frac{1}{2n}+\frac{n-1}{2n})=\frac{5n}{8}-1$, which is less than $\frac{n}{2}$ if and only if $n=2$. However if $n=2$ it is trivial that $x'$ is a pure $\epsilon$-Nash equilibrium.} However as we have $x_{\frac{n}{2}+1},\dots,x_n>\frac{n+1}{2n}$, this causes their utility to decrease. We have shown that under $x'$, no single agent can cause its utility to increase by more than $\epsilon$ by changing its report, and hence $x'$ is a $\epsilon$-Nash equilibrium. By exhaustion of cases, we have shown that a pure $\epsilon$-Nash equilibrium always exists under $f^*_{2IFS}$ if $n$ is even.

\textbf{Case 2 ($n$ is odd):} By using identical reasoning to the case where $n$ is even, we can see that if $x_i\geq \frac{2i-1}{2n}$ for any $i\in \{1,\dots,\frac{n-1}{2}$, a pure $\epsilon$-Nash equilibrium exists. Specifically, if we let $j=\arg\min_{i\in [\frac{n-1}{2}]} \{x_i\geq \frac{2i-1}{2n}\}$, then the pure $\epsilon$-Nash equilibrium is $x'=(1,\dots,1)$ if $j=1$, and $x'=(\frac{1}{2n}-\epsilon,\dots,\frac{2(j-1)-1}{2n}-(j-1)\epsilon,1,\dots,1)$ if $j>1$. Furthermore, symmetric reasoning shows that if $x_i\leq \frac{2i-1}{2n}$ for any $i\in \{\frac{n+3}{2},\dots,n\}$, a pure $\epsilon$-Nash equilibrium exists. It remains to consider the case where $x_i< \frac{2i-1}{2n}$ for all $i\in \{1,\dots,\frac{n-1}{2}\}$ and $x_i> \frac{2i-1}{2n}$ for all $i\in \{\frac{n+3}{2},\dots,n\}$ (and $x_{\frac{n+1}{2}}$ can be anywhere).

\textbf{Subcase 2a:} Here we consider the subcase where $x_{\frac{n+1}{2}}\geq \frac{1}{2}$. We claim that $x'=(\frac{1}{2n}-\epsilon,\frac{3}{2n}-2\epsilon,\dots,\frac{n-2}{2n}-(\frac{n-1}{2})\epsilon,1,\frac{n+2}{2n}+(\frac{n+3}{2})\epsilon,\dots,1-\frac{1}{2n}+\epsilon)$ is a pure $\epsilon$-Nash equilibrium. Here the interval of feasible agent locations is $[\frac{n-1}{2n}-(\frac{n-1}{2})\epsilon,\frac{n+1}{2n}+(\frac{n+3}{2})\epsilon]$, and the facility is placed at the leftmost point of the interval, at $y'=\frac{n-1}{2n}-(\frac{n-1}{2})\epsilon$, as there is a majority of agents right of the interval. Any misreport by agent $1$ will cause the facility to be placed in the interval $[0,\frac{1}{n}-2\epsilon]$, which reduces their utility as $x_1<\frac{1}{2n}$. Agent $i\in \{2,\dots,\frac{n-3}{2}\}$ can only cause the facility to be placed in $[\frac{i-1}{n}-(i-1)\epsilon,\frac{i}{n}-(i+1)\epsilon]$, which reduces their utility. A symmetric argument can be applied to show that agents $\frac{n+5}{2},\dots,n$ cannot improve their utility by misreporting. Agent $\frac{n-1}{2}$ (who reports $x'_{\frac{n-1}{2}}=\frac{n-2}{2n}-(\frac{n-1}{2})\epsilon$) can improve its utility by strictly less than $\epsilon$ by changing its report to $x''_{\frac{n-1}{2}}=x'_{\frac{n-1}{2}}+\epsilon_1$, where $\epsilon_1<\epsilon$. Similar to Subcase 1a, if $\epsilon_1\geq \epsilon$, the agent's utility will be decreased. As the facility takes the leftmost feasible point under $x'$, agent $\frac{n+3}{2}$ can only change the facility location to the rightmost feasible point (at $\frac{n+3}{2n}+(\frac{n+5}{2})\epsilon$), by misreporting such that there is a majority of agents left of the facility location. This reduces the agent's utility. Finally, it is easy to see that agent $\frac{n+1}{2}$ cannot improve its utility by misreporting: the infeasible regions under $x'$ are a result of the other agent reports, and hence agent $\frac{n+1}{2}$ cannot cause the facility to be placed outside of the feasible interval. As $x_{\frac{n+1}{2}}\geq \frac{1}{2}$, the leftmost point of the interval is its most optimal facility placement over its possible reports. Therefore $x'$ is a pure $\epsilon$-Nash equilibrium as no agent can improve its utility by more than $\epsilon$ by misreporting, and hence a pure $\epsilon$-Nash equilibrium exists for this subcase.

\textbf{Subcase 2b:} In this subcase, $x_{\frac{n+1}{2}}< \frac{1}{2}$ (and $x_i< \frac{2i-1}{2n}$ for all $i\in \{1,\dots,\frac{n-1}{2}\}$ and $x_i> \frac{2i-1}{2n}$ for all $i\in \{\frac{n+3}{2},\dots,n\}$). By using a symmetric argument as in Subcase 2a, $x'=(\frac{1}{2n}-\epsilon,\frac{3}{2n}-2\epsilon,\dots,\frac{n-2}{2n}-(\frac{n-1}{2})\epsilon,0,\frac{n+2}{2n}+(\frac{n+3}{2})\epsilon,\dots,1-\frac{1}{2n}+\epsilon)$, where the facility is placed at $\frac{n+1}{2n}+(\frac{n+3}{2})\epsilon$, is a pure $\epsilon$-Nash equilibrium.

In this proof, we have provided a pure $\epsilon$-Nash equilibrium for every possible location profile, and hence by exhaustion of cases, a $\epsilon$-pure Nash equilibrium always exists for $f^*_{2IFS}$.
\end{proof}
\section{Proof of Theorem~\ref{thm: nasheq2UFS}}
\begin{customthm}{\ref{thm: nasheq2UFS}}
For any $\epsilon>0$, a pure $\epsilon$-Nash equilibrium always exists for $f^*_{2UFS}$.
\end{customthm}
\begin{proof}
In the proof of Theorem~\ref{thm: nasheq2IFS}, we divided all possible agent location profiles into several subcases, and provide a pure $\epsilon$-Nash equilibrium for each subcase. We claim for every subcase, the pure $\epsilon$-Nash equilibrium we describe in the proof of Theorem~\ref{thm: nasheq2IFS} is also a $\epsilon$-Nash equilibrium for $f^*_{2UFS}$. First, we remark that every given pure $\epsilon$-Nash equilibrium has the same facility placement under $f^*_{2UFS}$. This can be seen as every pure $\epsilon$-Nash equilibrium has the $n$ agents reporting $n$ distinct locations, with the exception of the equilibria where multiple agents report $0$ or $1$. Under these equilibria, the facility takes the rightmost (resp. leftmost) location of the feasible interval, so the change to a $2-$UFS constraint does not affect the facility placement. 

A simple case by case analysis shows that for each pure $\epsilon$-Nash equilibrium $x'$ described in the proof of Theorem~\ref{thm: nasheq2IFS}, no agent can improve its utility by more than $\epsilon$ by changing its report under $f^*_{2UFS}$. The same arguments hold verbatim for $f^*_{2UFS}$ and the constraint of $2-$UFS, even when accounting for agents being able to change their report to the same location as another agent's report (to widen the `infeasible' interval around the report). We give the following intuition: if agent $i$ makes such a report change from $x_i'$, this either causes the facility to move to some location near $x_i'$ which has consequently become feasible, or it `pushes' the facility towards $x_i$ (such as if the facility takes the leftmost feasible location under $x'$ and agent $i$ changes its report from $x_i'=1$ to the rightmost reported location left of the facility). Hence every possible agent location profile has a pure $\epsilon$-Nash equilibrium under $f^*_{2UFS}$.
\end{proof}
\section{Proof of Theorem~\ref{thm:PoSupper}}
\begin{customthm}{\ref{thm:PoSupper}}
For $f^*_{2IFS}$ and $f^*_{2UFS}$, taking the limit $\epsilon \rightarrow 0$, the $\epsilon$-price of stability is at most 2.
\end{customthm}
\begin{proof}
We prove this theorem by iterating through each subcase in the proof of Theorem~\ref{thm: nasheq2IFS}, which gives the facility placement under a pure $\epsilon$-Nash equilibrium for types of location profiles. For each subcase, we consider the welfare ratio between the utilitarian welfare corresponding to the optimal facility placement when agents report truthfully, and the utilitarian welfare corresponding to the given $\epsilon$-equilibrium facility placement. We give the agent locations which maximize this welfare ratio and compute the ratio, showing that it is upper bounded by 2. While the proof is given for $f^*_{2IFS}$, it also holds verbatim for $f^*_{2UFS}$, as the set of pure $\epsilon$-Nash equilibria for $f^*_{2IFS}$ are also equilibria for $f^*_{2UFS}$.

We first note that $f^*_{2IFS}$ cannot place the facility in $(0,\frac{1}{2n})$ as it would imply that there are $0$ agents left of the facility, but then placing the facility at $0$ would result in strictly higher utilitarian welfare. Also note that if $f^*_{2IFS}(x)=0$, then the welfare ratio is $1$ as $x'=(1,\dots,1)$ is a pure $\epsilon$-Nash equilibrium.

\textbf{Case 1 ($n$ is even)}: 

\textbf{Subcase 1a:} Here we suppose that $x_i\geq \frac{2i-1}{2n}$ for some $i\in\{1,\dots,\frac{n}{2}\}$. Let $j=\arg\min_{i\in [\frac{n}{2}]}\{x_i\geq \frac{2i-1}{2n}\}$. Also suppose that in the symmetric side of the location profile, there is no equivalent $k\in \{\frac{n}{2}+1,\dots,n\}$ such that $x_k \leq \frac{2k-1}{2n}$ and $k>n-j+1$, else we would be able to `flip' the profile and consider strictly smaller $j$. As explained in the proof of Theorem~\ref{thm: nasheq2IFS}, there exists a pure $\epsilon$-Nash equilibrium that places the facility at $\frac{j-1}{n}-(j-1)\epsilon$. Taking the limit $\epsilon \rightarrow 0$, our welfare ratio here is
$$\frac{\sum_{i=1}^n|x_i-f^*_{2IFS}(x)|}{\sum_{i=1}^n|x_i-\frac{j-1}{n}|}.$$ The numerator is maximized by setting $x_1=0$ and $f^*_{2IFS}(x)=\frac{1}{2n}$, with the remaining agents $x_2\dots x_n$ located right of this facility placement. Under our $\epsilon$-Nash equilibrium placement, the facility is located between agents $x_{j-1}$ and $x_j$, thus to maximize the welfare ratio we take the agent locations $x_2\dots x_{j-1}$ to be as rightmost as possible (under the constraint that $j=\arg\min_{i\in [\frac{n}{2}]}\{x_i\geq \frac{2i-1}{2n}\}$). Thus we have $x_2,\dots,x_{j-1}=\frac{3}{2n}-\delta,\dots,\frac{2j-3}{2n}-(j-2)\delta$ for some sufficiently small $\delta>0$. By our initial assumptions, we also take $x_n\dots x_{n-j+2}=\frac{2n-1}{2n}+\delta,\dots,\frac{2(n-j+2)-1}{2n}+(j-1)\delta$, which are the lowest values that do not violate our assumptions (by taking the lowest values, we are maximizing the welfare ratio w.r.t. these locations). Finally, we minimize the values of $x_j\dots x_{n-j+1}$ to $\frac{2j-1}{2n},\dots,\frac{2j-1}{2n}+(n-2j+2)\delta$. Thus under the location profile that maximizes the welfare ratio under our $\epsilon$-Nash equilibrium, the utilitarian welfare corresponding to $f^*_{2IFS}(x)=\frac{1}{2n}$ is
\begin{align*}
&\frac{1}{2n}+(\frac{1}{n}+\dots+\frac{j-2}{n})+(n-2j+2)(\frac{j-1}{n})+(\frac{n-1}{n}+\dots+\frac{n-j+1}{n})\\
&=\frac{1}{2n}+\frac{(j-1)(j-2)}{2n}+\frac{2(n-2j+2)(j-1)}{2n}+\frac{(j-1)(2n-j)}{2n}\\
&=\frac{-4j^2+4nj+6j-4n-1}{2n},
\end{align*}
and (taking $\epsilon\rightarrow 0$) the utilitarian welfare corresponding to $f^*_{2IFS}(x')=\frac{j-1}{n}+(j-1)\epsilon$ is
\begin{align*}
&\frac{j-1}{n}+(\frac{2j-5}{2n}+\dots+\frac{1}{2n})+(\frac{2n-4j+5}{2n}+\dots+\frac{2n-j}{2n})\\
&=\frac{4j-4}{4n}+\frac{(j-2)(2j-4)}{4n}+\frac{(j-1)(4n-5j+5)}{4n}\\
&=\frac{-3j^2+4jn+6j-4n-1}{4n}.
\end{align*}
Dividing these terms gives us the maximum welfare ratio of $$\frac{8nj+12j-8j^2-8n-1}{4nj+6j-3j^2-4n-1}=2-\frac{2j^2-1}{-3j^2+4jn+6j-4n-1}.$$ The derivative of $-3j^2+4jn+6j-4n-1$ with respect to $j$ is $4n+6-6j$, which is positive for all $j\in [1,\frac{n}{2}]$, thus it is monotonic increasing with respect to $j$ and has a constrained minimum at $j=1$. As $-3j^2+4jn+6j-4n-1$ is positive at $j=1$, we see that the fraction $\frac{2j^2-1}{-3j^2+4jn+6j-4n-1}$ is positive for all $n$ and $j\in \{1,\dots,\frac{n}{2}\}$. Therefore in this subcase, the maximum welfare ratio is upper bounded by $2$ for all even $n\geq 2$. and $j\in \{1,\dots,\frac{n}{2}\}$.

\textbf{Subcase 1b:} Here we suppose that $x_i>\frac{2i-1}{2n}$ for $i\in \{\frac{n}{2}+2,\dots,n\}$. We also suppose that $x_i<\frac{2i-1}{2n}$ for $i\in \{1,\dots,\frac{n}{2}\}$, as otherwise the maximum welfare ratio is already covered in Subcase 1a. For this set of equilibria, we have $f^*_{2IFS}(x')\in (x_{\frac{n}{2}},x_{\frac{n}{2}+2})$, so as in the previous subcase, we note that the following agent locations maximize the welfare ratio: $x_1=0$ (which results in $f^*_{2IFS}(x)=\frac{1}{2n}$, maximizing the numerator), $x_2,\dots,x_{\frac{n}{2}}=\frac{3}{2n}-\delta,\dots,\frac{n-1}{2n}-(\frac{n}{2}-1)\delta$ (to maximize the numerator and minimize the denominator), and $x_{\frac{n}{2}+2},\dots,x_n=\frac{n+3}{2n}+(\frac{n}{2}-1)\delta,\dots,\frac{2n-1}{2n}+\delta$ (as subtracting the numerator and denominator by the same quantity increases the ratio), for some sufficiently small $\delta>0$. The proof is divided into two further subcases depending on where $x_{\frac{n}{2}+1}$ is located.

\textbf{Subcase 1bi:} In this subcase we suppose that $x_{\frac{n}{2}+1}<\frac{1}{2}+\frac{1}{2n}$, and from the proof of Theorem~\ref{thm: nasheq2IFS}, we know that there is a pure $\epsilon$-Nash equilibrium under which $f^*_{2IFS}$ places the facility\footnote{While the proof of Theorem~\ref{thm: nasheq2IFS} gives the facility location in terms of $\epsilon$, we use $\delta$ in this proof for consistency within the proof. There is no difference as we take the limit $\delta \rightarrow 0$.} at $\frac{n+2}{2n}+(\frac{n}{2}-1)\delta$. To maximize the welfare ratio, we set $x_{\frac{n}{2}+1}$ to be as rightmost as possible, at $x_{\frac{n}{2}+1}=\frac{1}{2}+\frac{1}{2n}-\delta$. Setting $\delta \rightarrow 0$, the utilitarian welfare corresponding to $f^*_{2IFS}(x)=\frac{1}{2n}$ is
\begin{align*}
&\frac{1}{2n}+(\frac{2}{2n}+\dots+\frac{n-2}{2n})+\frac{1}{2}+(\frac{n+2}{2n}+\dots+\frac{2n-2}{2n})\\
&=\frac{1}{2n}+\frac{n-2}{8}+\frac{1}{2}+\frac{3(n-2)}{8}\\
&=\frac{2n^2-2n+2}{4n},
\end{align*}
and the utilitarian welfare corresponding to $f^*_{2IFS}(x')=\frac{n+2}{2n}+(\frac{n}{2}-1)\delta$ is
\begin{align*}
&\frac{n+2}{2n}+(\frac{n-1}{2n}+\dots+\frac{3}{2n})+\frac{1}{2n}+(\frac{1}{2n}+\dots+\frac{n-3}{2n})\\
&=\frac{n+2}{2n}+\frac{(n-2)(n+2)}{8n}+\frac{1}{2n}+\frac{(n-2)^2}{8n}\\
&=\frac{n^2+6}{4n}.
\end{align*}
Dividing these welfare values gives our maximum welfare ratio of $\frac{2n^2-n+2}{n^2+6}$, which has an upper bound of $2$ for all even $n\geq 2$.

\textbf{Subcase 1bii:} In this subcase we have $x_{\frac{n}{2}+1}\geq \frac{1}{2}+\frac{1}{2n}$ and an equilibrium facility placement at $f^*_{2IFS}(x)=\frac{1}{2}-(\frac{n}{2})\delta$. To maximize the welfare ratio we take $x_{\frac{n}{2}+1}= \frac{1}{2}+\frac{1}{2n}$, and remark that when taking $\delta \rightarrow 0$, this is the same location profile as in the proof of Lemma~\ref{lem: PoSlower} (with a welfare ratio of $2-\frac{6n}{2n^2+n+2}$. Thus we see that in this subcase, the maximum welfare ratio is upper bounded by $2$ for all even $n\geq 2$.

\textbf{Subcase 1c:} In this subcase, we have $x_i<\frac{2i-1}{2n}$ for all $i\in \{1,\dots,\frac{n}{2}\}$, $x_i>\frac{2i-1}{2n}$ for all $i\in \{\frac{n}{2}+1,\dots,n\}$, and an equilibrium facility placement of $\frac{1}{2}-(\frac{n}{2})\delta$. Under these constraints, the agent locations which maximize the welfare ratio are the same as in Subcase 1bii, except with $x_{\frac{n}{2}+1}= \frac{1}{2}+\frac{1}{2n}+\delta$. Taking $\delta \rightarrow 0$, we end up with the same welfare ratio, which is upper bounded by $2$ for all even $n\geq 2$.

\textbf{Subcase 2 ($n$ is odd):} Similar to Subcase 1a, we suppose that $x_i\geq \frac{2i-1}{2n}$ for some $i\in \{1,\dots,\frac{n-1}{2}\}$, and let $j=\arg\min_{i\in [\frac{n-1}{2}]}\{x_i\geq \frac{2i-1}{2n}\}$, which results in an equilibrium facility placement of $\frac{j-1}{n}$. We again suppose that in the symmetric side of the location profile, there is no equivalent $k\in \{\frac{n+3}{2},\dots,n\}$ such that $x_k\leq \frac{2k-1}{2n}$ and $k>n-j+1$. The same location profile in Subcase 1a can be applied to odd $n$, and also maximizes the welfare ratio for odd $n$. Therefore in this subcase, the maximum welfare ratio is upper bounded by $2$ for all even $n\geq 2$.

\textbf{Subcase 2a:} In this subcase, $x_i<\frac{2i-1}{2n}$ for all $i\in \{1,\dots,\frac{n-1}{2}\}$ and $x_i>\frac{2i-1}{2n}$ for all $i\in \{\frac{n+3}{2},\dots,n\}$. Also, $x_{\frac{n+1}{2}}\geq \frac{1}{2}$. As shown in the proof of Theorem~\ref{thm: nasheq2IFS}, there exists a pure $\epsilon$-Nash equilibrium where the facility is placed at $\frac{n-1}{2n}-(\frac{n-1}{2})\delta$, between agent locations $x_{\frac{n-1}{2}}$ and $x_{\frac{n+1}{2}}$. Thus to maximize the welfare ratio, we set $x_1=0$, $x_2,\dots,x_{\frac{n-1}{2}}=\frac{3}{2n}-\delta,\dots,\frac{n-2}{2n}-(\frac{n-3}{2})\delta$, $x_{\frac{n+1}{2}}=\frac{1}{2}$, and $x_{\frac{n+3}{2}},\dots,x_n=\frac{n+2}{2n}+(\frac{n-1}{2})\delta,\dots,\frac{2n-1}{2n}+\delta$, which results in $f^*_{2IFS}(x)=\frac{1}{2n}$. Taking the limit $\delta \rightarrow 0$, the utilitarian welfare corresponding to $f^*_{2IFS}(x)=\frac{1}{2n}$ is
\begin{align*}
&\frac{1}{2n}+(\frac{2}{2n}+\dots+\frac{n-3}{2n})+\frac{1}{2}-\frac{1}{2n}+(\frac{n+1}{2n}+\dots+\frac{2n-2}{2n})\\
&=\frac{1}{2}+\frac{(n-3)(n-1)}{8n}+\frac{(n-1)(3n-1)}{8n}\\
&=\frac{2n^2-2n+2}{4n},
\end{align*}
and the utilitarian welfare corresponding to $f^*_{2IFS}(x')=\frac{n-1}{2n}-(\frac{n-1}{2})\delta$ is
\begin{align*}
&\frac{n-1}{2n}+(\frac{n-4}{2n}+\dots+\frac{1}{2n})+\frac{1}{2n}+(\frac{3}{2n}+\dots+\frac{n}{2n})\\
&=\frac{4n}{8n}+\frac{(n-3)^2}{8n}+\frac{(n-1)(n+3)}{8n}\\
&=\frac{n^2+3}{4n}.
\end{align*}
Dividing these welfares gives us the maximum welfare ratio of $\frac{2n^2-2n+3}{n^2+3}$, which is upper bounded by $2$ for all odd $n\geq 3$.

\textbf{Subcase 2b:}In this subcase, we again have $x_i<\frac{2i-1}{2n}$ for all $i\in \{1,\dots,\frac{n-1}{2}\}$ and $x_i>\frac{2i-1}{2n}$ for all $i\in \{\frac{n+3}{2},\dots,n\}$. However we now have $x_{\frac{n+1}{2}}<\frac{1}{2}$, and an equilibrium facility placement of $\frac{n+1}{2n}+(\frac{n+3}{2n})\delta$. To maximize the welfare ratio we set $x_{\frac{n+1}{2}}=\frac{1}{2}-\delta$ and the remaining agent locations the same as in Subcase 2a. By taking the limit $\delta \rightarrow 0$, it is apparent that the utilitarian welfare corresponding to $f^*_{2IFS}(x)=\frac{1}{2n}$ is the same as in Subcase 2a, taking a value of $\frac{2n^2-2n+2}{4n}$. Furthermore, the utilitarian welfare corresponding to $f^*_{2IFS}(x')=\frac{n+1}{2n}+(\frac{n+3}{2n})\delta$ is also the same as in Subcase 2a. To see this, note that the only difference between this subcase and Subcase 2a is that the equilibrium facility placement has changed from $\frac{n-1}{2n}-(\frac{n-1}{2})\delta$ to $\frac{n+1}{2n}+(\frac{n+3}{2n})\delta$. As $\delta \rightarrow 0$, the distance between the facility and $x_{\frac{n+1}{2}}$ remains the same, and the utility lost by the agents at $x_{\frac{n+3}{2}},\dots,x_n$ is matched by the utility gained by the agents at $x_1,\dots,x_{\frac{n-1}{2}}$. Therefore in this subcase, the maximum welfare ratio is also $\frac{2n^2-2n+3}{n^2+3}$, which is upper bounded by $2$ for all odd $n\geq 3$.

By exhaustion of cases, we have shown for $f^*_{2IFS}$ that for every agent location profile, there exists for sufficiently small $\epsilon>0$, a pure $\epsilon$-Nash equilibrium corresponding to a facility location that gives at least half of the utilitarian welfare when agents report truthfully (i.e. the welfare ratio is upper bounded by $2$). In other words, for $f^*_{2IFS}$, as $\epsilon \rightarrow 0$, the $\epsilon$-price of stability is at most $2$.
\end{proof}
\section{Proof of Lemma~\ref{lem: 01opt}}
\begin{customlem}{\ref{lem: 01opt}}
Consider an arbitrary agent location profile $x$. For every $2-$IFS/UFS randomized mechanism that gives positive support to a facility placement between $0$ and $1$, there exists a $2-$IFS/UFS randomized mechanism that only gives positive support to a facility placement at $0$ or $1$ that leads to weakly higher expected utility for each agent.
\end{customlem}
\begin{proof}
Consider an arbitrary location profile $x=(x_1,\dots,x_n)$, and suppose for this profile that some ($2-$IFS/UFS) mechanism places the facility at location $c\in (0,1)$ with probability $p$. If instead the mechanism placed the facility at location $1$ with probability $cp$ and at location $0$ with probability $p-cp$, each agent's expected distance from the facility would increase. This can be seen as for any $x_i\leq c$,
\begin{align*}
cp(1-x_i)+(p-cp)x_i&=px_i-2cpx_i+cp\\
&=cp+px_i(1-c)-cpx_i\\
&\geq cp-cpx_i\\
&\geq p(c-x_i),
\end{align*}
and for any $x_j>c$,
\begin{align*}
cp(1-x_j)+(p-cp)x_j&=px_j+cp(1-x_j)-cpx_j\\
&\geq px_j-cpx_j\\
&\geq p(x_j-c).
\end{align*}
Therefore this modified mechanism also satisfies $2-$IFS/UFS and results in weakly higher expected utility for each agent than the original mechanism. By repeatedly applying this modification, any $2-$IFS/UFS mechanism that places positive probability on any location between $0$ and $1$ can be modified to a $2-$IFS/UFS mechanism that only places positive probability on locations $0$ and $1$.
\end{proof}
\section{Proof of Theorem~\ref{thm: mech2}}
\begin{customthm}{\ref{thm: mech2}}
Mechanism 2 satisfies $2-$UFS. 
\end{customthm}
\begin{proof}
Consider a coalition of $|S|$ agents at location $x_i$ and suppose there are $n_1$ agents in $[0,\frac{1}{2}]$ and $n_2$ agents in $(\frac{1}{2},1]$. The expected distance from the facility is
\begin{align*}
\mathbb{E}(d(y,x_i))&=\frac{2n_1n_2+n_2^2}{n_1^2+n_2^2+4n_1n_2}x_i+\frac{n_1^2+2n_1n_2}{n_1^2+n_2^2+4n_1n_2}(1-x_i)\\
&=\frac{n_1^2+2n_1n_2}{n_1^2+n_2^2+4n_1n_2}+x_i\left(\frac{n^2_2-n^2_1}{n_1^2+n_2^2+4n_1n_2}\right).
\end{align*}
Due to symmetry it suffices to only consider $x_i\in [0,\frac{1}{2}]$, and since $\mathbb{E}(d(y,x_i))$ is a linear function of $x$, we further restrict our attention to $x_i\in \{0,\frac{1}{2}\}$.

When $x_i=\frac{1}{2}$, $\mathbb{E}(d(y,x_i))=\frac{1}{2}$ and hence $2-$UFS is satisfied for any coalition of agents at $\frac{1}{2}$. When $x_i=0$,
\begin{align*}
\mathbb{E}(d(y,x_i))&=\frac{n_1^2+2n_1n_2}{n_1^2+n_2^2+4n_1n_2}\\
&=\frac{n_1(2n_1^2+6n_1n_2+4n_2^2)}{2(n_1^2+n_2^2+4n_1n_2)(n_1+n_2)}\\
&>\frac{n_1(n_1^2+4n_1n_2+n_2^2)}{2(n_1^2+n_2^2+4n_1n_2)(n_1+n_2)}\\
&=\frac{n_1}{2(n_1+n_2)}\\
&\ge \frac{|S|}{2n}.
\end{align*}
Therefore Mechanism 2 satisfies $2-$UFS.
\end{proof}

\section{Proof of Lemma~\ref{lem:2IFSRand}}
\begin{customlem}{\ref{lem:2IFSRand}}
\textbf{2-IFS Randomized mechanism} is optimal amongst all randomized mechanisms satisfying 2-IFS in expectation.
\end{customlem}
\begin{proof}
We first prove by cases that the mechanism satisfies 2-IFS in expectation. Recall that for the OFLP, the utilitarian welfare maximizing solution places the facility at 0 or 1.

\textbf{Case 1} ($\sum^n_{i=1}x_i=\frac{n}{2}$)

In this case, the facility locations of $0$ and $1$ are tied for maximizing utilitarian welfare, so placing the facility at either location with probability $\frac{1}{2}$ is optimal. The mechanism also satisfies 2-IFS in expectation as the expected distance of any agent from the facility is $\frac{1}{2}x_i+\frac{1}{2}(1-x_i)=\frac{1}{2}\geq \frac{1}{2n}$.

\textbf{Case 2} ($\sum_{i=1}^nx_i>\frac{n}{2}$)

In this case, the optimal facility location is $0$. It is trivial that placing the facility at $0$ with probability $1$ when $x_1\geq \frac{1}{2n}$ is optimal and satisfies 2-IFS, so we consider the subcase where $x_1<\frac{1}{2n}$.

We first show that the mechanism satisfies 2-IFS in this case. The expected distance between $x_1$ and the facility is 
\begin{align*}
&\frac{2n-1-2nx_1}{2n(1-2x_1)}x_1+\frac{1-2nx_1}{2n(1-2x_1)}(1-x_1)\\
&=\frac{2nx_1-x_1-2nx_1^2+1-x_1-2nx_1+2nx_1^2}{2n(1-2x_1)}\\
&=\frac{1}{2n}.
\end{align*}
2-IFS is therefore satisfied for agents $x_2,\dots,x_n$ as $\alpha\leq \frac{1}{2}$.

We now show for this case that the mechanism is optimal amongst all randomized mechanisms satisfying 2-IFS in expectation. By Lemma~\ref{lem: 01opt}, it suffices to only consider mechanisms that can only place the facility at $0$ or $1$.

Now consider the 2-IFS Randomized mechanism. If $\alpha$ were increased, the utilitarian welfare would decrease, and if $\alpha$ were decreased, $2-$IFS would be violated for $x_1$, hence $\alpha=\frac{1-2nx_1}{2n(1-2x_1)}$ is optimal and therefore the mechanism is optimal under the constraint of 2-IFS in expectation.

\textbf{Case 3} ($\sum_{i=1}^nx_i>\frac{n}{2}$)

This case is similar and symmetric to Case 2.
\end{proof}

\section{Proof of Theorem~\ref{thm:2IFSRand}}
\begin{customthm}{\ref{thm:2IFSRand}}
\textbf{2-IFS Randomized mechanism} has an approximation ratio of $\frac{12}{11}\approx 1.091$.
\end{customthm}
\begin{proof}
It suffices to consider the case where $\sum_{i=1}^nx_i>\frac{n}{2}$ and $x_1<\frac{1}{2n}$ as the case where $\sum_{i=1}^nx_i>\frac{n}{2}$ is symmetric, and the mechanism is optimal for the cases of $\sum^n_{i=1}x_i=\frac{n}{2}$, and $\sum_{i=1}^nx_i>\frac{n}{2}$ and $x_1\geq \frac{1}{2n}$.

The approximation ratio of the mechanism is
\allowdisplaybreaks
\begin{align*}
\max_{x\in X^n}\Bigg\{\frac{\phi^*(x)}{\phi_{2IFS} (x)}\Bigg\}&= \max_{x\in X^n}\frac{\sum_ix_i}{(1-\alpha)\sum_i x_i+\alpha \sum_i (1-x_i)}\\
&=\max_{x\in X^n}\frac{\sum_ix_i}{\frac{2n-1-2nx_1}{2n(1-2x_1)}\sum_ix_i+\frac{1-2nx_1}{2n(1-2x_1)}\sum_i(1-x_i)}\\
&=\max_{x\in X^n}\frac{\sum_ix_i}{\frac{1-2nx_1}{2(1-2x_1)}+\frac{n-1}{n(1-2x_1)}\sum_ix_i}\\
&=\max_{x\in X^n}\frac{1}{\frac{1-2nx_1}{2(1-2x_1)\sum_i x_i}+\frac{n-1}{n(1-2x_1)}}\\
&=\max_{x_1\in [0,\frac{1}{2n})}\frac{1}{\frac{1-2nx_1}{2(1-2x_1)(n-1+x_1)}+\frac{n-1}{n(1-2x_1)}}\\
&=\max_{x_1\in [0,\frac{1}{2n})}\frac{2n(1-2x_1)(n-1+x_1)}{n-2n^2x_1+2(n-1)(n-1+x_1)}.
\end{align*}
In the second last line, we substitute $x_2,\dots,x_n=1$ as the ratio is monotonic increasing with $\sum_ix_i$. Some optimization programming shows that when $n\geq 3$, the ratio is maximized when $x_1=0$. Substituting $x_1=0$ into the ratio gives

$$\frac{2n^2-2n}{2n^2-3n+2},$$
which has a derivative of $-\frac{2(n^2-4n+2)}{(2n^2-3n+2)^2}$. We therefore see our ratio has a maximum turning point at $x=2+\sqrt{2}$ and is monotonic decreasing after this point. For integer $n\geq 3$, the ratio is maximized when either $n=3$ or $n=4$, and the ratio is equal to $\frac{12}{11}\approx 1.091$ for both of these points.

We now consider the case where $n=2$ (and $x_1<\frac{1}{4}$). The ratio becomes
\begin{align*}
\max_{x_1\in [0,\frac{1}{4})}\frac{2-2x_1-4x_1^2}{2-3x_1}&=\frac{1}{9}(22-4\sqrt{10})\\
&\approx 1.039. 
\end{align*}
Therefore the mechanism's welfare ratio is maximized when $x_1=0$ and either $n=3$ or $n=4$, taking a value of $\frac{12}{11}\approx 1.091$.
\end{proof}

\section{Proof of Lemma~\ref{lem:2UFSRand}}
\begin{customlem}{\ref{lem:2IFSRand}}
\textbf{2-UFS Randomized mechanism} is optimal amongst all randomized mechanisms satisfying $2-$UFS in expectation.
\end{customlem}
\begin{proof}
Similar to the proof of Lemma~\ref{lem:2IFSRand}, we prove this statement by cases.

\textbf{Case 1 ($\sum_{i=1}^m |S_i|x_i=\frac{n}{2}$)}

When $\sum_{i=1}^m |S_i|x_i=\frac{n}{2}$, the facility locations of $0$ and $1$ are tied for maximizing utilitarian welfare, so placing the facility at either location with probability $\frac{1}{2}$ is optimal. The mechanism also satisfies $2-$UFS in expectation as the expected distance of any agent from the facility is $\frac{1}{2}x_i+\frac{1}{2}(1-x_i)=\frac{1}{2}$.

\textbf{Case 2 ($\sum_{i=1}^m |S_i|x_i>\frac{n}{2}$)}

Note that in this case the optimal facility location is $0$. We first show that the mechanism satisfies $2-$UFS in this case. For a group of agents $S_i$ at $x_i<\frac{1}{2}$, the expected distance from the facility is
\begin{align*}
\alpha(1-x_i)+x_i(1-\alpha)&\geq \alpha_i(1-x_i)+x_i(1-\alpha_i)\\
&=\frac{|S_i|-2nx_i}{2n(1-2x_i)}(1-x_i)+\frac{2n-2nx_i-|S_i|}{2n(1-2x_i)}x_i\\
&=\frac{|S_i|(1-2x_i)}{2n(1-2x_i)}=\frac{|S_i|}{2n}.
\end{align*}
By setting $|S_i|=n$, we also see that $\alpha\leq \frac{1}{2}$, hence $2-$UFS is satisfied for any group of agents at $x_j\geq \frac{1}{2}$.

We now show for this case that the mechanism is optimal amongst all randomized mechanisms satisfying 2-UFS in expectation. By Lemma~\ref{lem: 01opt}, it suffices to only consider mechanisms that can only place the facility at $0$ or $1$. Now under the $2-$UFS Randomized mechanism, increasing $\alpha$ would decrease the utilitarian welfare, and decreasing $\alpha$ would violate $2-$UFS for some group of agents. Hence the mechanism is optimal under the constraint of $2-$UFS in expectation.

\textbf{Case 3 ($\sum_{i=1}^m |S_i|x_i<\frac{n}{2}$)}

This case is similar and symmetric to Case 2.
\end{proof}

\section{Proof of Theorem~\ref{thm:2UFSRand}}
\begin{customthm}{\ref{thm:2UFSRand}}
\textbf{2-UFS Randomized mechanism} has an approximation ratio of $\frac{2}{7}(1+2\sqrt{2})\approx 1.09384$.
\end{customthm}
\begin{proof}
Without loss of generality we suppose that $\sum^n_{i=1}>\frac{n}{2}$. Let $j$ be the index of the group of agents corresponding to $\alpha$ (i.e. $\alpha_j=\max\{\alpha_1,\dots,\alpha_k\}$).

The approximation ratio of the mechanism is
\begin{align*}
\max_{x\in X^n}\Bigg\{\frac{\phi^*(x)}{\phi_{2UFS} (x)}\Bigg\}&= \max_{x\in X^n}\frac{\sum_ix_i}{(1-\alpha)\sum_i x_i+\alpha \sum_i (1-x_i)}\\
&=\max_{x\in X^n}\frac{\sum_ix_i}{\frac{2n-|S_j|-2nx_j}{2n(1-2x_j)}\sum_i x_i+\frac{|S_j|-2nx_j}{2n(1-2x_j)} \sum_i (1-x_i)}\\
&=\max_{x\in X^n}\frac{\sum_ix_i}{\frac{|S_j|-2nx_j}{2(1-2x_j)}+\frac{n-|S_j|}{n(1-2x_j)}\sum_ix_i}\\
&=\max_{x\in X^n}\frac{1}{\frac{|S_j|-2nx_j}{2(1-2x_j)\sum_ix_i}+\frac{n-|S_j|}{n(1-2x_j)}}\\
&=\max_{\substack{x_j\in [0,\frac{1}{2})\\|S_j|\in \{1,\dots,n-1\}}} \frac{1}{\frac{|S_j|-2nx_j}{2(1-2x_j)(n-|S_j|+|S_j|x_j)}+\frac{n-|S_j|}{n(1-2x_j)}}\\
&=\max_{\substack{x_j\in [0,\frac{1}{2})\\|S_j|\in \{1,\dots,n-1\}}}\scalebox{1.1}{$\frac{2n(1-2x_j)(n-|S_j|+|S_j|x_j)}{n|S_j|-2n^2x_j+2(n-|S_j|)(n-|S_j|+|S_j|x_j)}$}\\
&=\max_{\substack{x_j\in [0,\frac{1}{2})\\r\in (0,1)}}\frac{2(1-2x_j)(1-r+rx_j)}{r-2x_j+2(1-r)(1-r+rx_j)}.
\end{align*}
where $r=\frac{|S_j|}{n}$. Some optimization programming shows that this ratio is maximized at $r=1-\frac{1}{\sqrt{2}}$ and $x_j=0$, taking a value of $\frac{2}{7}(1+2\sqrt{2})$.
\end{proof}

 \section{Lemma~\ref{lemma:intersection} and Proof of Lemma~\ref{lemma:intersection}}
Here we introduce an auxiliary lemma which will be used in the proof of Theorem~\ref{exist2PF}.
\begin{lemma}\label{lemma:intersection}
				The intersection of $(I-B_i)$'s  is not empty.
			\end{lemma}
			\begin{proof}
				We consider two cases:

\noindent\textbf{Case 1:} Every open interval $B_i,  i\in[m]$ lies completely in interval $I=[0,1]$.  Hence, the boundary points of interval $I$ are not in each $B_i$'s  and therefore $\{0,1\}\in \cap_{i=1}^m(I-B_i)$.

\noindent\textbf{Case 2:} There exists an open interval which does not completely lie in the interval $I$. Assuming the points $x_1, x_2, ..., x_m$ are ordered from left to right, let $k$ be the smallest index such that $B_k$ does not completely lie in $I$. So either $0\in B_k$ or $1\in B_k$. Both end points $0$ and $1$ can  not  be in $B_k$, since in this case $|B_k|=\frac{|S_k|}{n}>1$. Let us assume  $0\in B_k$ and $1\notin B_k$, i.e, $1\in (I-B_k)$. Now if for every $i\in[m]$, we have $1\in (I-B_i)$, the lemma statement holds.  Now suppose there exists an open interval $B_j$ such that  $1\notin (I-B_j)$, i.e, $1\in B_j$.  Without loss of generality let $j$ be the largest index in which the aforementioned statement holds. We consider two subcases:

\textbf{Case 2a:} $B_k\cap B_j$ is not an empty set. In this case $[0,1]\subseteq B_k\cup B_j$ and $\frac{|S_k|}{n}+\frac{|S_j|}{n}>1$, which is impossible.

\textbf{Case 2b:} $B_k\cap B_j$ is an empty set. We consider the set  $A:=I-(B_k\cup B_j)$ and consider two cases:
						\begin{itemize}
							\item $A\subseteq \cup_{i=k+1}^{j-1} B_i$: in this case $[0,1]\subseteq \cup_{i=1}^{m} B_i$ and
							$\sum_{i=1}^m \frac{|S_i|}{n}>1$, which is impossible.
							\item  $A\nsubseteq \cup_{i=k+1}^{j-1} B_i$. Therefore there exists $y\in A$ such that $y\notin  \cup_{i=k+1}^{j-1} B_i$.  Also by the definition of $A$, $y\notin (B_k\cup B_j)$.  For every $1\leq i \leq k-1$, $B_i\subset B_k$ holds as $i$ is the smallest index such that its corresponding interval contains the point $0$. Similarly, we have for every $j+1 \leq i \leq m$, $B_i\subset B_j$. Hence $y\notin B_i$ for every $i\in[m]$, i.e, $y\in (I-B_i)$ for every $i\in[m]$ and $y\in \cap_{i=1}^m (I-B_i)$.
						\end{itemize}
			\end{proof}
\section{Proof of Theorem~\ref{exist2PF}}\label{app:exist2PF}
\begin{customthm}{\ref{exist2PF}}
A   2-PF solution always exists.
\end{customthm}
\begin{proof}
Suppose we have $m$ unique agent locations, i.e. $m$ groups of agents.  We prove the theorem  by induction on the number of groups. When all the agents have the same location, i.e, $m=1$, we can allocate the facility $y$ at the furthest boundary point. This trivially satisfies 2-PF.   Now, we assume for any $k$ groups of agents where $k\leq m$ that there exists a  2-PF solution, and we extend that for  $k=m+1$.

					Suppose we have $m+1$ groups of agents placed at points $c_1, c_2,\dots,c_{m+1}$, which are ordered such that $c_1\leq c_2\leq \dots \leq c_{m+1}$. Let $S_i$ denote the group of agents at $c_i$. We set an open interval $B_i$ with length $\frac{|S_i|}{n}$, around  each center $c_i$. To prove the theorem, we consider several cases.

\noindent \textbf{Case 1} (There is no overlap between any two intervals, i.e. $B_i\cap B_{j}=\emptyset \forall i,j\in\{1,2,\dots,m+1\}$). In Lemma~\ref{lemma:intersection}, we have shown that the intersection of $(I-B_i)$'s  is not empty, so there exists a point  $y$ outside of every interval $B_i$ where the facility can be placed. Here we define $r$ as the distance between centers $c_i$ and $c_j$. If  $r=0$, then 2-PF is satisfied since we have $d(y, c_i)\geq \frac{|S_i|}{2n}$. Otherwise, since all the intervals are disjoint, $r$ is larger than the sum of the `radii' of the intervals $B_i$ and $B_j$, so   $\frac{|S_i|}{2n}+\frac{|S_j|}{2n}-r<0$. Hence, $$d(y,c_k)\geq \frac{|S_i|}{2n}+\frac{|S_j|}{2n}-r \quad \text{ for }k=i,j,$$ and thus $y$ satisfies the 2-PF inequality.
						
\noindent \textbf{Case 2} (There exists at least two overlapping intervals, i.e., $\exists~~ i, j\in \{1,2,\dots,m+1\}$ s.t. $B_i\cap B_j\neq\emptyset$).
						
							\textbf{Case 2a} (There exists one interval that is contained within another interval, i.e., $\exists~~ i, j\in \{1,2,\dots,m+1\}~ \text{s.t.} B_i \subseteq B_j$).  
							Without loss of generality, we assume $B_{2}\subseteq B_{1}$. 
							Now we place all the agents at $c_2$ and $c_{1}$ together at $c_1$. We set a new interval $B^{'}$ centered at $c_1$ with length $\frac{|S_1|}{n}+\frac{|S_{2}|}{n}$, which is the summation of the lengths of $B_1$ and $B_{2}$. 
							Now, we have $m$ new groups of agents located at the centers $c_1, c_3,\dots,c_{m+1}$. By induction, a  2-PF solution $y$ exists. We claim that $y$ is also a  2-PF solution for $m+1$ groups of agents located $c_1, c_2, \dots, c_{m+1}$. Since $y$ is a  2-PF solution for the $m$ groups of agents located at $c_1, c_3,..., c_{m+1}$,  $y$ lies outside of every interval, satisfying the 2-PF inequality for $r=0$. 
							
							Since $y$ is outside the interval  $B'$,  $y$ is outside of intervals $B_1$ and $B_2$, therefore we have $d(y,c_1)\geq \frac{|S_1|}{2n}$ and  $d(y,c_2)\geq \frac{|S_2|}{2n}$. So $y$ satisfies 2-PF inequalities for $r=0$. We now set $r$ to be the distance between two centers $c_1$ and $c_2$ (i.e., $r=c_2-c_1$). Since $y$ is outside the interval $B'$ we have $d(y,c_1)\geq \frac{|S_1|}{2n}+\frac{|S_{2}|}{2n}\geq \frac{|S_1|}{2n}+\frac{|S_{2}|}{2n}-r $, so $c_1$ satisfies the PF inequality. To show that the agents at $c_2$ also satisfy the PF inequality, we consider $2$ subcases:
			
							\begin{itemize}
								\item \textbf{Case 2ai} (Interval $B_2$ does not contain center $c_1$ (see Figure \ref{fig9})). We denote $a:=(c_2-c_1)-\frac{|S_2|}{2n}$ as the distance between $c_1$ and the left boundary of $B_2$, $b:=\frac{|S_2|}{n}$ as the length of $B_2$ and $c:=c_1+\frac{|S_1|}{2n}-c_2-\frac{|S_2|}{2n}$ as the distance between the right boundaries of $B_1$ and $B_2$. In this case, we have 
								\begin{align}\label{eq1}
									\frac{|S_1|}{2n}+\frac{|S_2|}{2n}-r=a+b+c+\frac{b}{2}-a-\frac{b}{2}=b+c.
								\end{align}
								
								Since $y$ is outside of interval $B'$, we have $d(y,c_2)\geq b+c$ and by (\ref{eq1}), $d(y,c_2)\geq \frac{|S_1|}{2n}+\frac{|S_2|}{2n}-r$ and thus the 2-PF inequality is satisfied.
								\begin{figure}[htp]
									\centering
									\begin{tikzpicture}
			\draw[] (-4,0) -- (4,0); 
			\draw [thick] (-4,0.3)  -- (-4,-0.3);  
			\draw [thick] (4,0.3)  -- (4,-0.3); 
		
			\node at (-4, -0.5) {$0$}; 
			\node at (4, -0.5) {$1$}; 
			\node at (-0.5, -0.25) {$c_1$}; 
			\node at (-0.5,0) [circle,fill,inner sep=1.5pt]{};
			
			\node at (1, -0.25) {$c_2$}; 
			\node at (1,0) [circle,fill,inner sep=1.5pt]{};
			
			\node at (1, 0.5) {$B_2$}; 
		\draw[] (0.5,0.25) -- (1.5,0.25);
		\draw [thick] (0.5,0.1)  -- (0.5,0.4);  
			\draw [thick] (1.5,0.1)  -- (1.5,0.4);

			\node at (-0.5, 1) {$B_1$}; 
			\draw[] (-3,0.75) -- (2,0.75);
		\draw [thick] (-3,0.6)  -- (-3,0.9); 
			\draw [thick] (2,0.6)  -- (2,0.9); 
			
			\node at (-0.5, 1.5) {$B'$}; 
			\draw[] (-3.5,1.25) -- (2.5,1.25);
			\draw [thick] (-3.5,1.1)  -- (-3.5,1.4); 
			\draw [thick] (2.5,1.1)  -- (2.5,1.4);
			
			\draw [densely dotted, thick] (-0.5,0) -- (-0.5,-1);
			\draw [densely dotted, thick] (0.5,0.25) -- (0.5,-1);
			\draw[>=triangle 45, <->] (-0.5,-1) -- (0.5,-1);
			\node at (0, -0.75) {$a$};
			
			\draw [densely dotted, thick] (1.5,0.25) -- (1.5,-1);
			\draw[>=triangle 45, <->] (0.5,-1) -- (1.5,-1);
			\node at (1, -0.75) {$b$};
			
			\draw [densely dotted, thick] (2,0.75) -- (2,-1);
			\draw[>=triangle 45, <->] (1.5,-1) -- (2,-1);
			\node at (1.75, -0.75) {$c$};
			
			\draw [densely dotted, thick] (2.5,1.25) -- (2.5,-1);
			\draw[>=triangle 45, <->] (2,-1) -- (2.5,-1);
			\node at (2.25, -0.65) {$\frac{b}{2}$};
		\end{tikzpicture}
									\caption[Example illustrating Case 2ai in the proof of Theorem~\ref{exist2PF}.]{$B_1$ and $B'$ are open intervals centered at $c_1$ and $B_2$ is an open interval centered at $c_2$. The variables $a$, $b$, and $c$ denote the distances between the boundaries of the intervals and the centers.}
									\label{fig9}
								\end{figure}
								\item \textbf{Case 2aii} (Interval $B_2$ contains center $c_1$ (see Figure \ref{fig10})). In this case, the range $r$ is smaller than the length of $B_2$. We denote $a:=c_2-c_1$ as the distance between the two centers, $b:=\frac{|S_2|}{n}$ as the length of $B_2$ and $c:=c_1+\frac{|S_1|}{2n}-c_2-\frac{|S_2|}{2n}$ as the distance between the right boundaries of $B_1$ and $B_2$. We have
								\begin{align}\label{eq2}
									\frac{|S_1|}{2n}+\frac{|S_2|}{2n}-r=a+\frac{b}{2}+c+\frac{b}{2}-a=b+c.
								\end{align}
								Since $y$ is outside of interval $B'$, we have $d(y,c_2)\geq b+c$ and by (\ref{eq2}), $d(y,c_2)\geq \frac{|S_1|}{2n}+\frac{|S_2|}{2n}-r$ and thus the 2-PF inequality is satisfied.
								\begin{figure}[htp]
									\centering
									\begin{tikzpicture}
			\draw[] (-4,0) -- (4,0); 
			\draw [thick] (-4,0.3)  -- (-4,-0.3);  
			\draw [thick] (4,0.3)  -- (4,-0.3); 
		
			\node at (-4, -0.5) {$0$}; 
			\node at (4, -0.5) {$1$}; 
			\node at (0, -0.25) {$c_1$}; 
			\node at (0,0) [circle,fill,inner sep=1.5pt]{};
			
			\node at (0.5, -0.25) {$c_2$}; 
			\node at (0.5,0) [circle,fill,inner sep=1.5pt]{};
			
			\node at (0.5, 0.5) {$B_2$}; 
		\draw[] (-0.25,0.25) -- (1.25,0.25);
		\draw [thick] (-0.25,0.1)  -- (-0.25,0.4);  
			\draw [thick] (1.25,0.1)  -- (1.25,0.4);

			\node at (0, 1) {$B_1$}; 
			\draw[] (-2.5,0.75) -- (2.5,0.75);
		\draw [thick] (-2.5,0.6)  -- (-2.5,0.9); 
			\draw [thick] (2.5,0.6)  -- (2.5,0.9); 
			
			\node at (0, 1.5) {$B'$}; 
			\draw[] (-3.25,1.25) -- (3.25,1.25);
			\draw [thick] (-3.25,1.1)  -- (-3.25,1.4); 
			\draw [thick] (3.25,1.1)  -- (3.25,1.4);
			
			\draw [densely dotted, thick] (0,0) -- (0,-1);
			\draw [densely dotted, thick] (0.5,0) -- (0.5,-1);
			\draw[>=triangle 45, <->] (0,-1) -- (0.5,-1);
			\node at (0.25, -0.75) {$a$};
			
			\draw [densely dotted, thick] (1.25,0.25) -- (1.25,-1);
			\draw[>=triangle 45, <->] (0.5,-1) -- (1.25,-1);
			\node at (0.875, -0.65) {$\frac{b}{2}$};
			
			\draw [densely dotted, thick] (2.5,0.75) -- (2.5,-1);
			\draw[>=triangle 45, <->] (1.25,-1) -- (2.5,-1);
			\node at (1.875, -0.75) {$c$};
			
			\draw [densely dotted, thick] (3.25,1.25) -- (3.25,-1);
			\draw[>=triangle 45, <->] (2.5,-1) -- (3.25,-1);
			\node at (2.875, -0.65) {$\frac{b}{2}$};
		\end{tikzpicture}
									\caption[Example illustrating Case 2aii in the proof of Theorem~\ref{exist2PF}.]{$B_1$ and $B'$ are open intervals centered at $c_1$ and $B_2$ is an open interval centered at $c_2$. The variables $a$, $b$, and $c$ denote the distances between the boundaries of the intervals and the centers.}
									\label{fig10}
								\end{figure}
							\end{itemize}

\textbf{Case 2b} (There exist two overlapping intervals, i.e. $\exists~~ i, j\in \{1,2,\dots,m+1\}$ s.t. $B_i\cap B_j\neq \emptyset$, but they are not contained within each other.) Without loss of generality, we assume $B_{1}\cap B_{2}\neq \emptyset$. Consider the line segment that connects the left border of $B_1$ to the right border of  interval $B_{2}$ and denote the midpoint of this line as $c'_1:=\frac{1}{2}(c_1-\frac{|S_1|}{2n}+c_2+\frac{|S_2|}{2n})$. We move all the agents at $c_1$ and $c_{2}$ to point $c'_1$ and  set a new interval $B'$ centered at $c'_1$ with length $\frac{|S_1|}{n}+\frac{|S_{2}|}{n}$, which is the summation of the lengths of $B_1$ and $B_{2}$. Similar to the previous subcase, by our inductive assumption there exists a  2-PF solution $y$ for the new  $m$ groups of agents  placed at  $c'_1, c_3,..., c_{m}, c_{m+1}$. Now we claim that $y$ is a  2-PF solution for $c_1, c_2, ..., c_{m+1}$ as well. Since $y$ is outside of the interval  $B'$,  $y$ is also outside of the intervals $B_1$ and $B_2$, therefore $d(y,c_1)\geq \frac{|S_1|}{2n}$ and  $d(y,c_2)\geq \frac{|S_2|}{2n}$. So $y$ satisfies the 2-PF inequalities for $r=0$. We now set $r$ to be the distance between two centers $c_1$ and $c_2$ (i.e., $r=c_2-c_1$) and consider 3 cases which depend on whether the intersections of the intervals also contain a center. In each case we show that $\frac{|S_1|}{2n}+\frac{|S_2|}{2n}-r$ is equal to the length of intersection of $B_1$ and $B_2$ which denote as $|intersection|$ (and hence the 2-PF inequality is satisfied).

								\begin{itemize}
									\item \textbf{Case 2bi} (The intersection between   $B_1$ and $B_2$ does not contain centers $c_1$ and $c_2$ (see Figure \ref{fig6})). We denote $a:=c_2-\frac{|S_2|}{2n}-c_1$ as the distance between $c_1$ and the left boundary of $B_2$, $b:=c_1+\frac{|S_1|}{2n}-(c_2-\frac{|S_2|}{2n})$ as the length of the intersection, and $c:=c_2-\frac{|S_1|}{2n}-c_1$ as the distance between $c_2$ and the right boundary of $B_1$. Here the length of the intersection is $\frac{|S_1|}{2n}+\frac{|S_2|}{2n}-r=a+b+b+c-a-b-c=b$.
									 Since $y$ is placed outside of interval $B'$, it is outside of intervals $B_1$ and $B_2$. We have $d(y,c_i)\geq \frac{|S_i|}{2n}\geq b=\frac{|S_1|}{2}+\frac{|S_2|}{2}-r$  and  thus 2-PF is satisfied.
					\begin{figure}[htp]
										\centering
										\begin{tikzpicture}
			\draw[] (-4,0) -- (4,0); 
			\draw [thick] (-4,0.3)  -- (-4,-0.3);  
			\draw [thick] (4,0.3)  -- (4,-0.3); 
		
			\node at (-4, -0.5) {$0$}; 
			\node at (4, -0.5) {$1$}; 
			\node at (-1, -0.25) {$c_1$}; 
			\node at (-1,0) [circle,fill,inner sep=1.5pt]{};
			
			\node at (1, -0.25) {$c_2$}; 
			\node at (1,0) [circle,fill,inner sep=1.5pt]{};
			
			\node at (1, 0.5) {$B_2$}; 
		\draw[] (0,0.25) -- (2,0.25);
		\draw [thick] (0,0.1)  -- (0,0.4);  
			\draw [thick] (2,0.1)  -- (2,0.4);

			\node at (-1, 1) {$B_1$}; 
			\draw[] (-2.5,0.75) -- (0.5,0.75);
		\draw [thick] (-2.5,0.6)  -- (-2.5,0.9); 
			\draw [thick] (0.5,0.6)  -- (0.5,0.9); 
			
			\node at (-0.25, 1.5) {$B'$}; 
			\draw[] (-2.75,1.25) -- (2.25,1.25);
			\draw [thick] (-2.75,1.1)  -- (-2.75,1.4); 
			\draw [thick] (2.25,1.1)  -- (2.25,1.4);
			
			\draw [densely dotted, thick] (-1,0) -- (-1,-1);
			\draw [densely dotted, thick] (0,0.2) -- (0,-1);
			\draw[>=triangle 45, <->] (-1,-1) -- (0,-1);
			\node at (-0.5, -0.75) {$a$};
			
			\draw [densely dotted, thick] (0.5,0.75) -- (0.5,-1);
			\draw[>=triangle 45, <->] (0,-1) -- (0.5,-1);
			\node at (0.25, -0.75) {$b$};
			
			\draw [densely dotted, thick] (1,0) -- (1,-1);
			\draw[>=triangle 45, <->] (0.5,-1) -- (1,-1);
			\node at (0.75, -0.75) {$c$};
			
			\draw [densely dotted, thick] (2.25,1.25) -- (2.25,-1);
			\draw [densely dotted, thick] (2,0.25) -- (2,-1);
			\draw[>=triangle 45, <->] (2,-1) -- (2.25,-1);
			\node at (2.125, -0.65) {$\frac{b}{2}$};
		\end{tikzpicture}
										\caption[Example illustrating Case 2bi in the proof of Theorem~\ref{exist2PF}.]{$B_1$, $B_2$ and $B'$ are open intervals centered at $c_1$, $c_2$ and $c'$ respectively. The terms $a$, $b$, and $c$ denote the distances between the boundaries of the intervals and the centers.}
										\label{fig6}
									\end{figure}
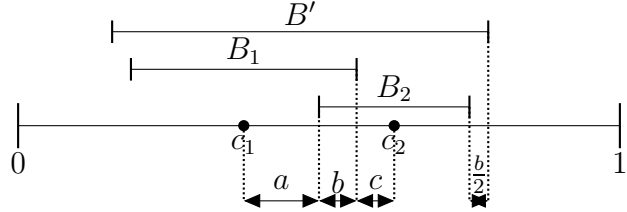

									\item \textbf{Case 2bii} (The intersection between $B_1$ and $B_2$ contains one of the centers $c_1$ and $c_2$ (see Figure \ref{fig7})). Without loss of generality suppose the contained center is $c_2$. We denote $a:=c_2-\frac{|S_2|}{2n}-c_1$ as the distance between $c_1$ and the left boundary of $B_2$, $b:=\frac{|S_2|}{2n}$ as the distance from $c_2$ to an endpoint of $B_2$, and $c:=c_2-\frac{|S_1|}{2n}-c_1$ as the distance between $c_2$ and the right boundary of $B_1$. Here the length of the intersection is $\frac{|S_1|}{2n}+\frac{|S_2|}{2n}-r=a+b+c+b-a-b=b+c$.
									 
									Since $y$ is outside of $B_1$, we have
									\begin{align*}
									d(y,c_1)&\geq \frac{|S_1|}{2n}= a+b+c\geq b+c\\
									&=\frac{|S_1|}{2n}+\frac{|S_2|}{2n}-r.
									\end{align*}
									Now we show $d(y,c_2)$ satisfies the 2-PF inequality. Recall $B'$ is an interval centered at the midpoint of the left boundary of $B_1$ and right boundary of $B_2$, with radius $\frac{|S_1|}{2n}+\frac{|S_2|}{2n}$. Also, $y$ lies outside of $B'$, so the right boundary of $B_2$ has a distance of length $\frac{b+c}{2}$ from the right boundary of $B'$. We therefore have
									$$d(y,c_2)\geq \frac{|S_2|}{2n}+\frac{b+c}{2}=\frac{b}{2}+\frac{b+c}{2}$$ 
									Also, since  $b\geq c$, we have
									$$b+\frac{b+c}{2}\geq b+c=\frac{|S_1|}{2n}+\frac{|S_2|}{2n}-r.$$

										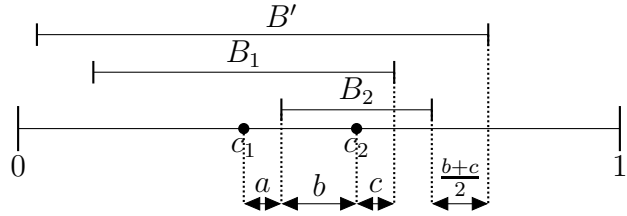
\begin{figure}[htp]
										\centering
										\begin{tikzpicture}
			\draw[] (-4,0) -- (4,0); 
			\draw [thick] (-4,0.3)  -- (-4,-0.3);  
			\draw [thick] (4,0.3)  -- (4,-0.3); 
		
			\node at (-4, -0.5) {$0$}; 
			\node at (4, -0.5) {$1$}; 

				\node at (0.5, -0.25) {$c_2$}; 
			\node at (0.5,0) [circle,fill,inner sep=1.5pt]{};
			\node at (0.5, 0.5) {$B_2$}; 
		\draw[] (-0.5,0.25) -- (1.5,0.25);
		\draw [thick] (-0.5,0.1)  -- (-0.5,0.4);  
			\draw [thick] (1.5,0.1)  -- (1.5,0.4);

		\node at (-1, -0.25) {$c_1$}; 
			\node at (-1,0) [circle,fill,inner sep=1.5pt]{};
			\node at (-1, 1) {$B_1$}; 
			\draw[] (-3,0.75) -- (1,0.75);
		\draw [thick] (-3,0.6)  -- (-3,0.9); 
			\draw [thick] (1,0.6)  -- (1,0.9); 
			
			\node at (-0.5, 1.5) {$B'$}; 
			\draw[] (-3.75,1.25) -- (2.25,1.25);
			\draw [thick] (-3.75,1.1)  -- (-3.75,1.4); 
			\draw [thick] (2.25,1.1)  -- (2.25,1.4);
			
			\draw [densely dotted, thick] (-1,0) -- (-1,-1);
			\draw [densely dotted, thick] (-0.5,0.2) -- (-0.5,-1);
			\draw[>=triangle 45, <->] (-1,-1) -- (-0.5,-1);
			\node at (-0.75, -0.75) {$a$};
			
			\draw [densely dotted, thick] (0.5,0) -- (0.5,-1);
			\draw[>=triangle 45, <->] (-0.5,-1) -- (0.5,-1);
			\node at (0, -0.75) {$b$};
			
			\draw [densely dotted, thick] (1,0.75) -- (1,-1);
			\draw[>=triangle 45, <->] (0.5,-1) -- (1,-1);
			\node at (0.75, -0.75) {$c$};
			
			\draw [densely dotted, thick] (2.25,1.25) -- (2.25,-1);
			\draw [densely dotted, thick] (1.5,0.25) -- (1.5,-1);
			\draw[>=triangle 45, <->] (1.5,-1) -- (2.25,-1);
			\node at (1.875, -0.65) {$\frac{b+c}{2}$};
		\end{tikzpicture}
										\caption[Example illustrating Case 2bii in the proof of Theorem~\ref{exist2PF}.]{$B_1$, $B_2$ and $B'$ are open intervals centered at $c_1$, $c_2$ and $c'$ respectively. The terms $a$, $b$, and $c$ denote the distances between the boundaries of the intervals and the centers.}
										\label{fig7}
									\end{figure}
									
									\item \textbf{Case 2biii} (The intersection between $B_1$ and $B_2$ contains both centers  $c_1$ and $c_2$ (see Figure \ref{fig8})). We denote $a:=c_2-\frac{|S_2|}{2n}-c_1$ as the distance between $c_1$ and the left boundary of $B_2$, $b:=c_2-c_1$ as the distance between the two centers, and $c:=c_2-\frac{|S_1|}{2n}-c_1$ as the distance between $c_2$ and the right boundary of $B_1$. The length of the intersection is
									$\frac{|S_1|}{2n}+\frac{|S_2|}{2n}-r=b+c+b+a-b=a+b+c$.
									In this case, $2(b+c)$ is the length of interval $B_1$ and $2(a+b)$ is the length of interval $B_2$. Since the intervals $B_1$ and $B_2$ are not contained within each other, we have $b+c> a$ and $a+b> c$ (see Figure \ref{fig8}). Like before, $B'$ is an interval with a center at the midpoint of the left boundary of $B_1$ and right boundary of $B_2$ and length $\frac{|S_1|}{n}+\frac{|S_2|}{n}$, and $y$ lies outside of this interval, so the boundaries of each interval $B_1$ and $B_2$ have a distance of length $\frac{a+b+c}{2}$ with the boundary of $B'$. We therefore have 
									
									\begin{align*}
									d(y,c_1)&\geq b+c+\frac{a+b+c}{2}\\
									&\geq b+c+a=\frac{|S_1|}{2n}+\frac{|S_2|}{2n}-r.
									\end{align*}
									
									Similarly, we can show that $y$ satisfies the 2-PF inequality for center $c_2$. 
									
\begin{align*}									
									d(y,c_2)&\geq a+b+\frac{a+b+c}{2}\\
									&\geq a+b+c=\frac{|S_1|}{2n}+\frac{|S_2|}{2n}-r.
									\end{align*}
							
Hence the facility placement of $y$ satisfies the 2-PF inequalities.

									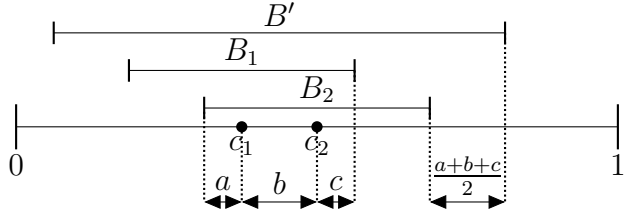
\begin{figure}[htp]
										\centering
										\begin{tikzpicture}
			\draw[] (-4,0) -- (4,0); 
			\draw [thick] (-4,0.3)  -- (-4,-0.3);  
			\draw [thick] (4,0.3)  -- (4,-0.3); 
		
			\node at (-4, -0.5) {$0$}; 
			\node at (4, -0.5) {$1$}; 

				\node at (0, -0.25) {$c_2$}; 
			\node at (0,0) [circle,fill,inner sep=1.5pt]{};
			\node at (0, 0.5) {$B_2$}; 
		\draw[] (-1.5,0.25) -- (1.5,0.25);
		\draw [thick] (-1.5,0.1)  -- (-1.5,0.4);  
			\draw [thick] (1.5,0.1)  -- (1.5,0.4);

		\node at (-1, -0.25) {$c_1$}; 
			\node at (-1,0) [circle,fill,inner sep=1.5pt]{};
			\node at (-1, 1) {$B_1$}; 
			\draw[] (-2.5,0.75) -- (0.5,0.75);
		\draw [thick] (-2.5,0.6)  -- (-2.5,0.9); 
			\draw [thick] (0.5,0.6)  -- (0.5,0.9); 
			
			\node at (-0.5, 1.5) {$B'$}; 
			\draw[] (-3.5,1.25) -- (2.5,1.25);
			\draw [thick] (-3.5,1.1)  -- (-3.5,1.4); 
			\draw [thick] (2.5,1.1)  -- (2.5,1.4);
			
			\draw [densely dotted, thick] (-1,0) -- (-1,-1);
			\draw [densely dotted, thick] (-1.5,0.2) -- (-1.5,-1);
			\draw[>=triangle 45, <->] (-1,-1) -- (-1.5,-1);
			\node at (-1.25, -0.75) {$a$};
			
			\draw [densely dotted, thick] (0,0) -- (0,-1);
			\draw[>=triangle 45, <->] (-1,-1) -- (0,-1);
			\node at (-0.5, -0.75) {$b$};
			
			\draw [densely dotted, thick] (0.5,0.75) -- (0.5,-1);
			\draw[>=triangle 45, <->] (0.5,-1) -- (0,-1);
			\node at (0.25, -0.75) {$c$};
			
			\draw [densely dotted, thick] (2.5,1.25) -- (2.5,-1);
			\draw [densely dotted, thick] (1.5,0.25) -- (1.5,-1);
			\draw[>=triangle 45, <->] (1.5,-1) -- (2.5,-1);
			\node at (2, -0.65) {$\frac{a+b+c}{2}$};
		\end{tikzpicture}
										\caption[Example illustrating Case 2biii in the proof of Theorem~\ref{exist2PF}.]{$B_1$, $B_2$ and $B'$ are open intervals centered at $c_1$, $c_2$ and $c'$ respectively. The terms $a$, $b$, and $c$ denote the distances between the boundaries of the intervals and the centers.}
										\label{fig8}
									\end{figure}
								\end{itemize}
					By exhaustion of cases, we have proven the inductive statement, and thus a facility placement that satisfies 2-PF always exists.
\end{proof}

\section{Proof of Theorem~\ref{thm: hybrid}}
\begin{customthm}{\ref{thm: hybrid}}
Under the hybrid model, a H-UFS solution always exists.
\end{customthm}
\begin{proof}
Let $n_C$ denote the number of classic agents and $n_O$ denote the number of obnoxious agents (such that $n_C+n_O=n$). Consider an arbitrary agent location profile with $m$ unique classic agent locations $x_1,\dots,x_m$, and suppose they are ordered such that $x_1\leq \dots x_m$. We first focus on the region of feasible facility locations pertaining to the classic agents' distance constraints. For $i\in [m]$, let $S_{C_i}$ denote the group of classic agents at location $x_i$, and construct a closed interval with center $x_i$ and radius $1-\frac{|S_{C_i}|}{n}$: $B_i=\{z|d(z,x_i)\leq 1-\frac{|S_{C_i}|}{n}\}$. By definition, the intersection of the closed intervals $\cap_{i\in [m]}B_i$ denotes the (continuous) region of feasible facility locations pertaining to the classic agents' distance constraints. From \citep{ALLW21}, we know this region is non-empty.

We show that the length of the feasible region $\cap_{i\in [m]}B_i$ is at least $\frac{n-n_C}{n}$ by iteratively transforming the agent location profile to one where all classic agents are at $0$ or $1$, and showing that each transformation weakly decreases the feasible region length. For each ball $B_i$, there is a (possibly empty) left interval of infeasible points $L_i$ and a (possibly empty) right interval of infeasible points $R_i$ such that $B_i=[0,1]-L_i-R_i$. We denote the left infeasible interval of points as $L_i=[0,x_i-(1-\frac{|S_{C_i}|}{n}))$ if $x_i-(1-\frac{|S_{C_i}|}{n})>0$ and as $L_i=\emptyset$ otherwise. Similarly, we denote the right infeasible interval of points as $R_i=(x_i+(1-\frac{|S_{C_i}|}{n}),1]$. The union of left infeasible intervals is therefore $\cup_{i\in [m]}L_i=[0,\max_{i\in [m]} x_i-(1-\frac{|S_{C_i}|}{n}))$ if there exists a nonempty $L_i$, and is empty otherwise. The union of right infeasible intervals is $\cup_{i\in [m]}R_i=(\min_{i\in [m]}x_i+(1-\frac{|S_{C_i}|}{n}),1]$ if there exists a nonempty $R_i$, and is empty otherwise. Therefore the length of the feasible region is $$\min\left\{\min_{i\in [m]}\left\{x_i+(1-\frac{|S_{C_i}|}{n})\right\},1\right\} - \max\left\{0,\max_{i\in [m]} \left\{x_i-(1-\frac{|S_{C_i}|}{n})\right\}\right\}.$$

By symmetry, we suppose without loss of generality that $$\min_{i\in [m]}\left\{x_i+(1-\frac{|S_{C_i}|}{n})\right\}\leq 1.$$ We consider the following transformation. Let $j$ correspond to the agent location with the largest right infeasible interval, i.e. $$j:=\arg\min_{i\in [m]}\left\{x_i+(1-\frac{|S_{C_i}|}{n})\right\},$$ and we have $x_j<1$. If $x_j>0$, move all agents at $x_j$ to $0$. We show that this transformation weakly decreases the feasible region length: in $\min_{i\in [m]}\left\{x_i+(1-\frac{|S_{C_i}|}{n})\right\}$, $x_i$ decreases to $0$ and $|S_{C_i}|$ weakly increases (it strictly increases if there are already agents at $0$). Furthermore, the $\max_{i\in [m]} \left\{x_i-(1-\frac{|S_{C_i}|}{n})\right\}$ term is unaffected unless the shifted group of agents originally corresponded to the maximum value, in which case the $x_i$ term decreases by at most the length of the shift, and the $(1-\frac{|S_{C_i}|}{n})$ term weakly decreases as $|S_{C_i}|$ weakly increases. Therefore this transformation weakly decreases the feasible region length. Now the agent location with the largest right infeasible interval is $0$, so our feasible region length is $$(1-\frac{|S_{C_{j'}}|}{n}) - \max\left\{0,\max_{i\in [m]} \left\{x_i-(1-\frac{|S_{C_i}|}{n})\right\}\right\},$$
where $j'$ corresponds to the location $x_{j'}=0$. If the right boundary of the largest left infeasible interval $\max_{i\in [m]} \left\{x_i-(1-\frac{|S_{C_i}|}{n})\right\}$ corresponds to the agents at $x_{j'}=0$, then it is at most $0$, and the feasible region length is $(1-\frac{|S_{C_{j'}}|}{n})$ which is at least $\frac{n-n_C}{n}$. We now suppose that this is not the case.

We have $$\max_{i\in [m]} \left\{x_i-(1-\frac{|S_{C_i}|}{n})\right\} \leq \max_{i\in [m]} \left\{1-(1-\frac{|S_{C_i}|}{n})\right\},$$ so the feasible region length is at least $$(1-\frac{|S_{C_{j'}}|}{n}) - \max\left\{0,\max_{i\in [m]} \left\{\frac{|S_{C_i}|}{n}\right\}\right\}.$$ Since $\arg\max_{i\in [m]} \left\{\frac{|S_{C_i}|}{n}\right\}\neq j'$, the feasible region length is at least $$1-\frac{|S_{C_{j'}}|+|S_{C_k}|}{n}\geq \frac{n-n_C}{n}$$ where $k\neq j'$.

We now know the length of the (continuous) feasible region corresponding to the classic agents is at least $\frac{n-n_C}{n}=\frac{n_O}{n}$. We now consider the obnoxious agents. Suppose we construct an open interval of radius $\frac{|S_{O_i}|}{2n}$ around each group of obnoxious agents at the same location. Any location within one of these open intervals is infeasible. The sum of interval lengths is $\frac{n_O}{n}$, so using a similar argument to that in the proof of Proposition~\ref{optUFS}, we see that a feasible solution with respect to both the classic agents' and obnoxious agents' distance inequalities always exists.
\end{proof}
\end{document}